\documentclass[a4paper,USenglish,cleveref, autoref]{lipics-v2019}

\title{On Properties of Network Systems}

\author{Evangelos Kipouridis}{BARC, University of Copenhagen, Universitetsparken 1, Copenhagen, Denmark}{kipouridis@di.ku.dk}{}{}

\author{Kostas Tsichlas}{School of Informatics, Aristotle University of Thessaloniki, Greece}{tsichlas@csd.auth.gr}{}{}

\authorrunning{E. Kipouridis and K. Tsichlas}

\Copyright{Evangelos Kipouridis and Kostas Tsichlas}



\category{}

\relatedversion{}

\nolinenumbers 

\hideLIPIcs  

\EventLongTitle{47th International Colloquium on Automata, Languages and Programming (ICALP 2020)}
\EventShortTitle{ICALP 2020}
\EventAcronym{ICALP}
\EventYear{2020}
\EventDate{}
\EventLocation{}
\EventLogo{}
\SeriesVolume{}
\ArticleNo{}

\newtheorem{prop}{Property}

\begin{document}

\maketitle

\begin{abstract}
The apparent disconnection between the microscopic and the macroscopic is a major issue in the understanding of complex systems. To this extent, we study properties of local rules that change the topology of a network, network systems as named by Wolfram, and barely touch on the expressive power of this model. There has been extensive research on this topic in various fields but almost in all cases (with minimum exceptions) the focus is either on modeling and experimental verification (complex systems) or in computability and complexity (communication networks) where rigorous proofs apply. Although there have been (notably not many) attempts to bridge these two general viewpoints, further effort is needed especially in transferring the computational viewpoint to the complex systems field. To be more precise, there may be some merit in looking at such networks from an algorithmic perspective with all the relevant analysis tools. We hope that this paper is a step towards this direction and will constitute an incentive for further research.

In this paper we look at the behavior of network systems when different types of local rules are applied on them. We prove convergence for a very general class of local rules, and provide guarantees on the speed of convergence for an important subclass of this class. We also study more general rules, and show that they do not converge. Our counterexamples resolve an open question of (Zhang, Wang, Wang, Zhou, KDD - 2009) as well, concerning whether a certain process converges. In addition, we provide a certain member of the aforementioned converging class that, when applied on a graph, efficiently computes its $k$-core and its $(k-1)$-crust giving hints on the expressive power of such a model. In the same spirit, we further provide local rules that solve the single source shortest path problem. Finally, we show the universality of our network system, by providing a local rule under which it is Turing-Complete. 
\end{abstract}

\newpage

\section{Introduction} \label{sec:intro}

There is an increasing interest on the interplay between the microscopic and the macroscopic in terms of emergent behavior. This is a crucial point for a better understanding of complex systems. The most striking examples come from biological systems that seem to form macroscopic structures out of local interactions between simpler structures on all levels of organization. For example, Physarum Polycephalum (a slime mold) has been shown \cite{journals/n/NakagakiYT00} to be able to solve computational problems such as the shortest path by simple local interactions. The underlying common characteristic of these systems is the emergent behavior in the macroscopic level out of simple local interactions at the microscopic level. Motivated by the plethora of such examples, we delve deeper into this field by studying repeated applications of local rules on a network and provide hints with respect to the expressive power of this model from a computer science perspective. 

There is a huge literature related to how the view is changed when the time dimension is introduced in the network. The incorporation of time in the network results in what is called in the literature dynamic networks \footnote{Henceforth the term network and graph refer to the same object.}, adaptive networks, time-varying networks, evolving networks and temporal networks that essentially refer to the same idea\footnote{Henceforth, we will use the term \textit{dynamic} network for consistency.}. This is mainly because dynamic networks have been the object of different scientific fields like sociology \cite{10.1086/428716}, physics \cite{doi:10.1142/S0219525911003050}, ecology \cite{rgd2018}, computer science \cite{DBLP:journals/cacm/MichailS18}, engineering (e.g., robotics, complex system design, etc.) \cite{saidani2004}, etc. However, in general, one can discern between two distinct viewpoints without excluding overlappingness: 

\begin{itemize}
 \item  \textbf{complex systems viewpoint:} (physics, sociology, ecology, etc.) the main focus is on modelling (differential/difference equations, random boolean networks, iterative maps, artificial neural networks and cellular automata - see \cite{Sayama2009}) and qualitative analysis (by means of mean field approximations, bifurcation analysis etc.). The main questions here are of qualitative nature and include phase transitions, complexity of system behavior, etc. Rigorous analysis is not frequent and simulation is the main tool for providing results.
 
 \item \textbf{computational viewpoint:} (mainly computer science and communications) the main focus is on the computational capabilities (computability/complexity) of dynamic networks in various settings and with different assumptions. The  main approach in computer science is based on rigorous proofs while in communications is based on experimental findings.

\end{itemize}
In a nutshell, the former asks what is the long/short term behavior of the system and its properties while the latter asks what can we compute and how fast. We provide more information as to these different viewpoints in Section \ref{ssec:related}. However, a small discussion on the computational viewpoint is imperative. The main approaches in this viewpoint is either to take the dynamic network for granted or reason w.r.t. computability/complexity about the specified local operations. In the first approach, (see for example \cite{Aggarwal:2014:ENA:2620784.2601412}) the source of dynamicity in the network is completely ignored and for the purposes of the analysis of the dynamics the effects of time on the network are assumed to be adversarial or stochastic, that is interactions are either worst-case or random. In the second approach, the local rules are specified and the questions asked are related to feasibility of computation and complexity to decide feasibility (e.g., \cite{ss2019,DBLP:conf/icalp/MichailSS17}). 

In the computational viewpoint what is missing is a programming language for local rules in a network and everything that such an achievement would entail. As discussed in the related work section and to the best of our knowledge, no such goal has been set in general in the computer science discipline. To accomplish this, we need to pinpoint the process that drives such a physical system (e.g., how Physarum Polycephalum actually evolves) or more generally program the network by means of local rules in order to get a target emergent behavior. This is one of the holy grails of complex systems theory and we think that computer science has a say about it. Note, that this approach is not data-driven since knowledge is required about the inner-workings of the network. This is probably one of the reasons why such an approach was not so popular in computer science (apart from its inherent difficulty). We hope that this paper is a small step towards this direction not only with respect to its results but also with respect to raising the awareness of the computer science community towards this direction.

In this paper, as discussed in \ref{ssec:results}, we touch on dynamic networks whose topological evolution alone is affected by local rules/programs. We look at problems related to the following questions: What is the emergent behavior in a network when a local rule is applied repetitively? What are the necessary conditions for the dynamic process driven by the repeated applications of local rules to converge and what does convergence means in this case? What local rules and assumptions should we use in order to solve a problem on a dynamic network? We provide results that give answers to certain settings with respect to the preceding questions. To do that, we first define a general model, which we call a network system. In the exposition of our results we tend to avoid cumbersome notation in order to make the results easier to reach and understand, without however sacrificing correctness. 

In the next section we provide the description of the proposed model as well as state the results of this paper. In addition, we discuss related work and how our model and the respective results are positioned within this vast field.

\section{Preliminaries} \label{sec:preliminaries}

We start by providing some definitions and notation and then move to the discussion of the proposed model. Then we describe the results of the paper and finally we discuss related work.

\subsection{Definitions}\label{ssec:def}

Let $G=(V,E)$ be an undirected simple graph. Assume that the graph evolves over time (discrete time) based on a set of rules. We represent the graph at time $t$ by $G^{(t)}=(V^{(t)},E^{(t)})$, and $G^{(0)}=G$. We define the energy of the edge $e=(u,v)$ at time $t$ to be some function related to this edge, which is represented by $\mathcal{E}_{G^{(t)}}^{(t)}(e)$ or $\mathcal{E}_{G^{(t)}}^{(t)}(u,v)$. We write $\mathcal{E}^{(t)}(u,v)$ or $\mathcal{E}(u,v)$ when the graph and the time we are referring to are clear from the context. Finally, assume that $n^{(t)}=|V^{(t)}|$ and $m^{(t)}=|E^{(t)}|$. 

Let $N_G(u)$ be the set of all neighbors of node $u$ and let $d_G(u)$ be the degree of node $u$ in graph $G$. We define $\left| E^{(t)}(u,v) \right|$ to be the number of edges between $u$ and $v$ at time $t$ (either $0$ or $1$), and $\left| E(G[N_{G^{(t)}}(u) \cap N_{G^{(t)}}(v)]) \right|$ to be the number of edges between common neighbors of $u$ and $v$ at time $t$.

Let $f: \mathbb{N}^2\rightarrow \mathbb{R}$ be a continuous function having the following two properties: i) Non-decreasing, that is $f(x,y+\epsilon)\geq f(x,y)$ for $\epsilon>0$ (similarly $f(x+\epsilon,y)\geq f(x,y)$) and ii) Symmetric, $f(x,y)=f(y,x)$. The second property is related to the fact that we consider undirected graphs. We call these functions \textit{proper}.

\subsection{The Network System}
\label{ssec:Network_System}
Informally, the network system in its general form is a discrete-time dynamic network of agents. Each agent executes its own algorithm that changes the labels/states on the network as well as its topology in an arbitrary way based on some interaction graph that also changes with time. In particular, a configuration $\mathcal{C}^{(t)}$ of a Network System (NS henceforth) is defined as follows:

\begin{itemize}
\item $G^{(t)}=(V^{(t)},E^{(t)}):$ A network of nodes $V^{(t)}$ and connections $E^{(t)}$ between them at time $t$. This is the underlying network where the dynamic process(es) are performed.
\item $L_V^{(t)}:V^{(t)}\rightarrow S_V:$ A map - possibly depending on time - from the node set $V^{(t)}$ to the node state set $S_V$. This is a set of labels for the nodes.
\item $L_E^{(t)}:E^{(t)}\rightarrow S_E:$ A map - possibly depending on time - from the edge set $E^{(t)}$ to the edge state set $S_E$. This is a set of labels for the edges.
\item $K^{(t)}:V^{(t)}\rightarrow \mathbb{N}:$ A map - possibly depending on time - that attaches a natural number to each node. This corresponds to the name of a node in case nodes are different between each other. For example, if all nodes in the network are anonymous then this map is constant and so we cannot discern between different nodes. 
\item $C^{(t)}=(V^{(t)},H^{(t)}):$ An interaction graph of nodes $V^{(t)}$ and connections $H^{(t)}$ that dictates the allowable communication links between pairs of nodes in $G^{(t)}$.
\end{itemize} 

The states and the topology of the network system are updated by means of algorithms (rules) that are run within each node in each discrete time step. As such, the temporal dynamics of the NS can be formally defined by the pair $(\mathcal{A},\mathcal{C}^{(0)})$:

\begin{itemize}
\item $\mathcal{A}:$ An algorithm (or update rule) that is run on each node in $V^{t}$. For an arbitrary time instance $t$ assume a node $v\in V^{(t)}$ that executes algorithm $\mathcal{A}$. $\mathcal{A}$ has access to all labels of $v$, its edges and their corresponding labels. In addition, it has similar access to adjacent nodes in the interaction graph $C^{(t)}$. However, the execution of $\mathcal{A}$ by node $v$ can only alter the labels of $v$, decide on the existence of edges between $v$ and its adjacent nodes in $C^{(t)}$ and alter their labels as well\footnote{There is no explicit control w.r.t. conflicts on the change of labels of edges by two different nodes.}. Notice that although $\mathcal{A}$ is common to all nodes, heterogeneity is captured since $\mathcal{A}$ has access to the names of the nodes induced by the mapping $K$. 
\item $\mathcal{C}^{(0)}:$ The initial configuration of the NS.
\end{itemize}

The NS can either be seen as a model for dynamic networks (complex systems approach) or as a model of computation (computer science approach). The interaction graph makes NS very flexible. For example, when used as a model for a physical process, the interaction graph can encode basic physical properties of the process like locality of interaction. When NS is used as a model of computation, the interaction graph can encode the computing restrictions posed by the nodes as well as the communication restrictions posed by the network itself.

It may seem peculiar that the NS contains two graphs, the network of agents and the interaction graph. However, this is not only a generalization for making it a general framework, as argued in the following discussion about simulation of other models by NS. In fact, it allows the model to focus on the emergence mechanism casting aside details related to computability (e.g. computing capabilities of agents) as well as capture aspects of such mechanisms related to different timescales. For example, imagine a social network $G$ where edges correspond to long term relations (e.g., friends) while the interaction graph corresponds to contacts during a time period. It may be the case that an agent does not contact some of his friends during a time period or it may also be the case that she contacts other agents with whom she does not have long term relations (no edge in $G$). Relating $G$ with the interaction graph $C$ through local rules can be achieved, for example, based on Structural Balance theory \cite{10.2307/2572978} taking into account different timescales.

With the existence of the interaction graph we manage to capture social dynamics into two different timescales. This explicit division into two timescales can also help set up more realistic models on networks like virus propagation and policy making with respect to a potential pandemic, by forcing massively less contacts in the interaction graph as well as reduction (but no so massive) in the network $G$ (e.g., by closing down schools). One could just as well generalize to a hierarchy of interaction graphs, which models interdependent dynamics on multiple timescales. The interaction graph $C$ can also be used to model assumptions (e.g., stochasticity) w.r.t. the underlying network. In particular, modeling random encounters between agents not connected by an edge in a social network can be done in an explicit way in the interaction graph. In subsection~\ref{ssec:simulation} we elaborate on this point by looking at how NS can simulate other well-known models for dynamic networks.

\subsection{Our results} \label{ssec:results}

In this paper, we restrict the model to capture only deterministic dynamic processes on the topology of the network. In this sense, we assume that the node and edge state sets are empty, that is $S_V=S_E=\emptyset$ and thus the mappings $L_V^{(t)}$ and $L_E^{(t)}$ are irrelevant. This restricted NS cannot capture dynamic processes on the network that change states of the nodes/edges, such as virus propagation, rumor spreading, etc., and thus it cannot also capture cases where such a dynamic process affects the topology of the network as well. It can only capture dynamic processes of the network, that is dynamic processes that change the topology of the network based only on its current topology.

To further simplify exposition, we assume without loss of generality that $\forall t: V^{(t)}=V$, that is the node set does not change. Changes to nodes can be emulated by having all possible nodes in $V$ and simply not allowing any relevant edges in $C^{(t)}$. For example, the Barab{\'a}si--Albert model \cite{Barabasi509}, can be simulated by simply setting $\mathcal{A}$ to add an edge between two nodes in $G^{(t)}$ when there exists such an edge in the interaction graph, which in turn these edges in $C^{(t)}$ are specified based on the preferential-attachment mechanism. We further restrict the model by looking at a particular family of deterministic algorithms $\mathcal{A}$\footnote{In the Barab{\'a}si--Albert model the algorithm is deterministic and the preferential attachment mechanism which is probabilistic in nature is implemented in the interaction graph $C^{(t)}$.} inspired by \cite{DBLP:conf/kdd/ZhangWWZ09} as described below. Notice that as shown in Section~\ref{sec:turing}, even this restricted Network System is Turing-Complete. 

Our goal is to find whether the following algorithm $\mathcal{A}$ stabilizes, or to state with other terms whether the dynamic process converges. In the following algorithm, the edge $e^{(t)}$ is also used as a boolean variable. In particular, when $e^{(t)}=0$ it means that $e^{(t)}\notin E^{(t)}$ (the edge is \textit{non-existent}) while $e^{(t)}=1$ means that $e^{(t)}\in E^{(t)}$ (the edge is \textit{existent}). Let $\alpha$ and $\beta$ be two parameters that correspond to two thresholds, the lower and the upper threshold respectively.

\begin{quote}
Algorithm $\mathcal{A}$\footnote{Another way to express the problem is to let $P^{(t)}$ be the vector with the energy values at time $t$ and to let $\mathcal{C}_{\alpha,\beta}$ be the deterministic process that applies the two thresholds on the energy. Then, the above algorithm can be expressed as:
\[P^{(t)}=\mathcal{E}(G^{(t)})\]
\[G^{(t+1)} = \mathcal{C}_{\alpha,\beta}(P^{(t)},G^{(t)})\]} at node $v:$

\hspace{0.3cm} For each adjacent to $v$ edge $e^{(t)} \in C^{(t)}$ we compute the energy $\mathcal{E}(e^{(t)})$. 

\hspace{0.6cm} Three cases for $\mathcal{E}(e^{(t)})$

\hspace{0.9cm} 1. $\mathcal{E}(e^{(t)})<\alpha$: Then edge $e^{(t+1)}=0$ (non-existent).

\hspace{0.9cm} 2. $\alpha \leq \mathcal{E}(e^{(t)})<\beta$: Then edge $e^{(t+1)}=e^{(t)}$. 
 
\hspace{0.9cm} 3. $\mathcal{E}(e^{(t)})\geq \beta$: Then edge $e^{(t+1)}=1$ (existent). 

\end{quote}

This algorithm is repetitively applied in each node in the network in each time instance $t$. Notice that for each edge $e^{(t)} \not\in C^{(t)}$ in the network, it holds that edge $e^{(t+1)}=e^{(t)}$. The algorithm is the same for all nodes since it does not use $K^{(t)}$. Note that in this algorithm the user needs only to define the energy, and in fact in the rest of the paper this is what we are going to do with some notable exceptions. Note that when $\alpha=\beta$, then the system has in effect no memory and it is in a sense an initial value problem. When $\alpha<\beta$, then the system is equipped with memory, which is not explicit, since the status of edges in maintained from the previous time point in case the energy of the edge falls in the range $[\alpha,\beta)$. This extra memory of the network system renders its analysis more complicated as we will see in Section~\ref{sec:convergence}.

We say that the algorithm \textit{converges} or \textit{stabilizes} when $\exists t: G^{(t)}=G^{(t+1)}$, which means that no change has taken place on the network. The time for the algorithm to converge is equal to the minimum such $t$, if it exists.
A discussion is in order with respect to the stabilizing condition of this algorithm. The user, apart from defining the local rule, must also define $C^{(t)}$. This stabilizing condition makes the silent and plausible assumption that the user does not behave as an adversary to the algorithm but in fact she tries to improve on the algorithm by trying to guarantee convergence, among others. In this sense, there is no meaning from the side of the algorithm to make an iteration without changing the graph. For example, if $C^{(t)}$ is the null graph then the algorithm always stabilizes. Notice that the stabilizing condition can be also related to other conditions based on the specified goal of the dynamic process; for example, the process may stabilize as soon as a clique of a particular size occurs in the network.

The following lemma describes the stabilizing condition for the case where $C^{(t)}=C$, that is the interaction graph is constant and does not change. This choice of $C^{(t)}$ is imposed by the fact that careful choices of $C^{(t)}$ can make almost any algorithm loop for ever, while for this case ($C^{(t)}=C$) it stabilizes (an example is shown in \ref{sssec:dependence}). 

\begin{lemma}\label{lem:end_cond_general}
Given that the interaction graph is constant $(\forall t: C^{(t)}=C)$ and for some $t'$ it holds that $G^{(t'-1)}=G^{(t')}$, then for any $t>t'$ it holds that $G^{(t)}=G^{(t')}$.
\end{lemma}
\begin{proof}
Since $C^{(t)}=C$, the update rule applies in each time step only on these pairs of nodes defined by $C$. Let $G^{(t'-1)}=G^{(t')}$.  It follows that no change took place at time $t'$. Since the algorithm is deterministic and the interaction graph does not change, no change will happen at the time step $t'+1$ and the lemma follows.
\end{proof}

In the case where the update rule does not converge, then since $V$ stays the same throughout the process, there are finitely many graphs with that many nodes. Thus, there exists some time point $t$ where the network from a previous step (say $t'<t-1$) has been reached. That is, $G^{(t')} = G^{(t)}$. We say that the cycle size of this process is $t-t'$ for the maximum $t'$, since the graphs $G^{(t')}, G^{(t'+1)},...,G^{(t-1)}$ will periodically repeat.

In the following, in order to demonstrate the potential applicability of this restricted network system we provide a small example inspired by structural balance theory \cite{10.2307/2572978} of networks with friendship and enmity relations (see \cite{ANTAL2006130}). Assume that the network of agents corresponds to people (nodes) with friendship relations (edges). Each agent $v$ is defined by how nice she is $n(v)$, how extrovert she is $x(v)$ as well as by the set of her enemies $\mathcal{EN}(v)$. We assume no labels on nodes and edges since these properties are not time-dependent while they can also be represented by the interaction graph. We wish to design a model that captures how friendships change in this setting when enemies do not change\footnote{Note that no changes are made on the network w.r.t. enemies that is not exactly compatible with structural balance theory on  networks.} as well as when friendships are lost in case of very few common friends and friends are made in the opposite case.

In this respect we define the social dynamics as follows: (a) if two agents $u$ and $v$ are enemies then they never become friends (no edge connects them in $C^{(t)}$ for any $t$, and thus they never connect in $G^{(t)}$ too), (b) else if two agents $u$ and $v$ are not connected by an edge (they are not friends) but their distance is at most the sum of their extrovertedness then they interact - that is, if at time $t$ it holds that $1<dist(u,v)\leq x(u)+x(v)$ then there is an edge $(u,v)$ in the interaction graph $C^{(t)}$, (c) if two agents $u$ and $v$ interact and their total confidence - energy - (the sum of the confidence of both agents) is at least as big as the threshold $\beta$ then they become friends and an edge is established. We define the \textit{confidence} of an agent $v$ as $g(v)=n(v)\sum_{u\in N(v)}{n(u)}$, so that it doesn't only depend on their niceness, but also on the niceness of their friends, (d) ceasing a friendship is more interesting; not only should the confidence of the two agents become less than the threshold $\alpha$ but they should also not have a lot of common friends that could balance things out. The first criterion is captured by the energy definition, while the second is captured in the interaction graph. More specifically, if two agents are connected by an edge in $G^{(t)}$, then they are only connected by an edge in $C^{(t)}$ if their number of common friends is at most $\gamma$. This $\gamma$ parameter could even depend on the time $t$ and the specific agents, but for demonstration purposes we just assume it to be a global parameter of the system, say $\gamma=10$.

For each node $v$ her extrovertedness, her enemies and her friends affect the interaction graph $C^{(t)}$ for each $t$. Since the energy is a simple sum and thus it is proper, and since confidence $g$ is a degree-like function (see \ref{ssec:Local_Rules} for its definition) it follows that Theorem~\ref{thm:conv_only} applies and thus the social network stabilizes. Theorem~\ref{thm:conv_only} allows us to add any rules w.r.t. the interaction graph $C^{(t)}$ like imposing a maximum number of friends, allowing for additional random connections (to achieve long-range interaction), etc. Similarly, we can change the definition of confidence as well as the definition of energy and still prove convergence as far as the assumptions of Theorem~\ref{thm:conv_only} are valid. Finally, the interaction graph allows us to remove the assumption of permanence on enmity by allowing edges after certain conditions have been satisfied, thus dynamically changing the set $\mathcal{EN}(v)$ for each node $v$.

We are ready now to move to the exposition of the results related to the restricted version of NS. First, in Section~\ref{sec:min} to exhibit the expressiveness in terms of emergent behavior we provide a very simple rule that, when applied on a graph, efficiently computes its $\alpha$-core and its $(\alpha-1)$-crust \cite{DBLP:conf/gd/BatageljMZ99}. 
In Section~\ref{sec:convergence} we provide proofs of convergence for a general class of energies. We also provide guarantees on the speed of convergence for an important subclass of the aforementioned class. Then, we study more general definitions of energies and show that in this case the algorithm does not converge. Our counterexamples resolve an open question of Zhang et al. \cite{DBLP:conf/kdd/ZhangWWZ09} as well, concerning whether a certain dynamic process to enhance network communities converges. 
In Section~\ref{sec:shortest} we show that this restricted NS can solve two versions of the shortest path problem by providing appropriate definitions of energy. In Section~\ref{sec:turing} we prove that the network system with a particular local rule is Turing-Complete. Finally, we conclude the paper with a discussion in Section~\ref{sec:discussion}.

\subsection{Related Work} \label{ssec:related}

As previously argued, the main work on dynamic networks stems either from computer science or from complex systems and it is inherently interdisciplinary in nature. In the following we just touch on these two approaches since trying to analyze the vast literature that spans all scientific fields would be a fool's errand. We also highlight all results that are directly related to ours.

Starting from the computer science realm, a very nice review of the dynamic network domain is in \cite{DBLP:journals/cacm/MichailS18} where a discussion is made with respect to how the Computer Science community attacks the introduction of the time dimension on networks. They propose a partitioning of the current literature into three subareas: Population Protocols, Powerful Dynamic Distributed Systems and Temporal Graphs. 

Population Protocols were introduced in \cite{Angluin2006, DBLP:journals/dc/AngluinAER07} and are still an active research topic (e.g., \cite{DBLP:journals/tcs/ChatzigiannakisMNPS11, DBLP:journals/acta/EsparzaGLM17}), where anonymous agents with only a constant amount of memory available \textit{randomly} interact with each other and are able to compute functions that do not seem possible at first, like leader election. Population protocols are a typical example of a passive model, and are related to chemical reaction networks \cite{DBLP:journals/nc/DotyZ18}.  There are also many active models, like Programmable Matter \cite{DBLP:journals/corr/abs-1807-10461, DBLP:conf/icalp/MichailSS17, DBLP:journals/cacm/MichailS18}, which models any type of matter that can algorithmically change its physical properties. In this case the goal \textit{is to check whether a desired configuration can be reached}, and if so, to provide a \textit{method to achieve it as fast as possible}. What distinguishes such active models from passive ones is the fact that movements of agents are controlled by the algorithm. On the other hand, on passive models the movements occur by the environment, and the algorithm can only decide whether to accept them or not. A similar notion to population protocols are graph relabeling systems \cite{graphrelabelling99}, where one chooses a subgraph and changes it based on certain rules. These systems are usually applied on static graphs but they have also been applied to dynamic graphs as well \cite{10.1007/978-3-642-11476-2_11}. The focus of this work is to \textit{impose properties on the dynamic graphs so that a particular computation is possible}, where the dynamic graphs are assumed to be adversarial. In general, graph relabeling systems aim at the abstraction of communication whereas the main goal of the population protocols is the abstraction of the dynamicity of the network. Finally, another very recent model, called the Network Constructors model \cite{Michail2016,Michail2017NetworkCA}, studied what stable networks can be constructed by a population of finite-automata that interact \textit{randomly} like molecules in a well-mixed solution and can establish bonds with each other according to the rules of a common small protocol. Their goal is to come up with protocols that allow the generation of particular networks (like paths, stars, rings and more complex networks). Apart from the stochastic interaction another difference to our approach is that they focus on what can be done given the knowledge that each agent has of its environment (e.g., its knowledge of the degree) leading to certain interesting complexity results. Apart from stabilisation, they also provide complexity results with respect to termination of the protocol. There are also a few ad-hoc results that aim at particular problems using mainly simple rules likes averaging or plurality. For example, in \cite{10.5555.3039686.3039745} it is shown how a simple averaging rule can lead to the identification of communities within large classes of networks. 

One of the first models for Powerful Dynamic Distributed Systems was developed in \cite{O'Dell:2005:IDH:1080810.1080828} which in fact correspond to dynamic networks with Turing-complete nodes. The main questions asked are algorithmic in nature and usually make the assumption that the dynamic network is either adversarial or stochastic. Some authors have also considered particular dynamicity patterns either stochastic \cite{10.1007/978-3-540-70575-8_11,doi:10.1137/090756053} (usually termed as Markovian Evolving Graphs) or deterministic \cite{FLOCCHINI201353} (periodic). However, in most of the cases the usage of dynamicity patterns is used more for the sake of analysis rather than for the sake of modeling the dynamic process that changes the network itself.

Finally, part of the literature of computer science on dynamic networks is about the structural and algorithmic properties of dynamic graphs. To this end, researchers developed frameworks that lead towards a Temporal Graph Theory. Time-Varying Graphs (TVGs) \cite{doi:10.1080/17445760.2012.668546} are a general framework that integrates concepts, formalisms and results related to graphs that carry time information. Note that TVG do not dictate the process that change the graph, this process is exogenous to TVG. Similarly, stream graphs and stream links \cite{Latapy2018} is a more recent framework for modeling such interactions over time. In each case, the temporal graphs are either controlled by the environment, by an adversary or are stochastic. For various results in this setting see \cite{Michail2015}.

From the previous discussion, it is clear that in the computer science domain the network dynamics are: 

\begin{itemize}
 \item \textbf{Adversarial:} The dynamics are governed by an adversary that chooses where and when changes will take place in the network. Usually certain restriction are imposed in order to make the analysis amenable (e.g., connectivity assumptions).
 \item \textbf{Stochastic:} The dynamics are governed by a stochastic process where changes are randomly chosen based on a distribution (e.g., the preferential attachment model for generating random networks). Apart from the choice for the stochastic model (e.g. markovian) there are also restrictions to make the analysis amenable (e.g., connectivity assumptions).
\end{itemize}

In the study of complex systems, one of the tools used for modeling and analysis is cellular automata. Cellular automata use simple update rules that give rise to interesting patterns \cite{DBLP:journals/corr/ArrighiD16, DBLP:journals/corr/abs-1711-10920}. A great example is one of the most simple cellular automata, Rule 110, which has been shown by Cook to be Turing-Complete \cite{DBLP:journals/compsys/000104}. However, in \cite{Ilachinski2018} starting from cellular automata, the authors introduced Structurally Dynamic Cellular Automata (SDCA) that couples the topology with the local site 0/1 value configuration. SDCA can be seen as a special case of NS. They formalize this notion and move to an experimental qualitative analysis of its behaviour for various parameteres. They left as extensions purely structural CA models (no value configurations as it holds in the restricted NS studied in this paper), probabilistic transition rules and extensions to more realistic models that take into account other quantities (e.g., length of edges). 
Network automata \cite{10.1007/3-540-19444-4_10} are cellular automata that allow for arbitrary but static connections between cells (automata). There are also models for dynamic connections where one can model at the same time dynamics on the states of the nodes as well as dynamics of its topology. One such preliminary example is graph automata \cite{tkm02-graphrew} that allows for dynamic connections by changing states as well as topology based on simple graph rewriting rules. In \cite{Sayama2009}, a general framework is described for modeling complex networks with dynamics on (functional) and of (topological) the network that may as well be coupled through the use of graph grammars. Simulations were used to explore the behavior for certain rules and ranges of parameters in order to show its rich behavior. In addition, a model for coupling topology with functional dynamics was given in \cite{doi:10.1142/S0219525911003050}. Their formulation is similar to ours and they extend it to stochastic rules as well as to restricted topologies. They also provide examples of their model, termed Functional Network Automata (FNA), to exhibit its expressiveness in a biological system. All their results are qualitative based on experimentation. Various additional versions of such models have been introduced - e.g., for automata that allow for the introduction of new cells see \cite{amod_2019}.
Based on this discussion, one would expect that our model would also be named Network Automata.  Wolfram \cite{2002:NKS:513738} uses the term network system for network automata that change their topology. We prefer to use this term instead of network automata to highlight our different approach. Finally, similarly to us, Silk et al. \cite{journals/ieeetnse/SilkHG16} designed local rules that allow a network to reach a desired steady-state degree distribution with proofs based on differential equations under certain assumptions.

Work has been also done on the predictability of emergent phenomena on cellular automata. In \cite{PhysRevE.73.026203} the authors use coarse-graining in order to study the cellular automaton on different scales by reducing the degrees of freedom. In this way, they manage to approximate the complex process by a simpler one with similar characteristics, whenever this is possible of course. This was extended in \cite{costa2019coarse} for partitioned cellular automata that are more appropriate in simulating linear and nonlinear problems in physics. 

There are also many other results in different scientific fields, ranging from communication \cite{6157579} to ecology \cite{sciencebirdformation1970}, on the design of local rules that give rise to a desired global behavior and their properties.  To name just a few, Valentini et al. proposed a global-to-local design methodology to compose heterogeneous swarms for self-organized task allocation \cite{ValHamDor2016:techreport-002}. In a similar spirit to ours, the author in \cite{saidani2004} looked at modular robots as an evolving network with respect to their topology alone. The author defined a graph topodynamic which in fact is a local program common to all modules of the robot with a specific target at mind (it is reminiscent of programmable matter although the questions asked are different). A program is provided to turn a tree topology to a chain topology while leaving as an open problem whether this truly converges. A game theoretic approach has also been adopted by allowing the dynamics to be governed by the interaction between the agents (nodes) that are rationally behaved based on their utility function. Here, the notion of convergence/stability is related to Nash equilibria. For example, in \cite{10.1086/428716} the prisoner's dilemma has been applied on an adaptive network and extensive experimental evaluation has been performed. 
To the best of our knowledge, the main tool for getting results in the complex systems domain is simulation.

Finally, in a similar spirit, a vast literature exists on multi-agent network systems where a set of agents iteratively interact over a communication network in order to achieve a particular goal like consensus or the optimization of an objective function. This literature spans mainly the case where the communication network is fixed. However, there is a lot of work that has been done in the case where state dynamics (dynamics on the network) and network dynamics (dynamics of the network) exist at the same time. The analysis of these systems is either based on the assumption of uncoupled (or loose coupling) network and state dynamics or in their strong coupling by assumming certain properties (like connectivity or communication symmetry). For example, in \cite{doi:10.1137/18M1217681,8814834} the author proposes a sequential optimization framework to prove whether such a dynamical system admits a Lyapunov function that allows us to prove stability. Chazelle has introduced influence systems \cite{DBLP:journals/siamcomp/Chazelle15} as a general framework for multiagent dynamics. In this setting there is a set of agents, each of which is modelled by a point in $\mathcal{R}^d$, while a communication graph governs their interaction. The edges of the communication graph are time-dependent and its edges are determined by first-order sentences over the reals. In the case of the diffusive influence systems where agents move to the convex region defined by their neighbors, Chazelle manages to prove almost always periodic behaviour under random perturbations while in the bidirectional case it always converge to a fixed point. The main goal of this model is to capture the complexity emanating from different communication patterns keeping the algorithms for the local rules very simple.

\subsection{Relation to Other Models}\label{ssec:simulation}

Our model is a general and expressive framework for analyzing dynamics on graphs. This is strengthened by the fact that the NS can simulate various other models as we discuss informally and briefly in the following. The Barab{\'a}si--Albert model \cite{Barabasi509}, can be simulated by simply setting $\mathcal{A}$ to add an edge between two nodes in $G^{(t)}$ when there exists such an edge in the interaction graph, where in turn these edges in $C^{(t)}$ are specified based on the preferential-attachment mechanism which is stochastic. Similarly, the Watts-Strogatz model \cite{Watts1998Collective} can be simulated by starting with a regual ring lattice and then in each step set the appropriate edges stochastically in the interaction graph to rewire them (the local algorithm).

NS resembles Functional Network Automata (FNA) \cite{doi:10.1142/S0219525911003050} but generalize it, since it allows for different interactions in each time step due to the interaction graph, while FNA forces every possible pair of interactions to take place. In fact, we show in Section~\ref{sec:convergence} that for the particular family of local rules Network Automata as defined in \cite{doi:10.1142/S0219525911003050} always converge. Influence systems are also a special case of NS since in Influence Systems the agents are points in the $d$-dimensional Euclidean space $\mathcal{R}^d$ while in our case apart from such properties they also have interelationships captured by the network $G^{(t)}$. 
From the distributed computing perspective, population protocols \cite{Angluin2006} can also be simulated since in this case the network $G^{(t)}$ coincides with the interaction graph $C^{(t)}$ and can either be a clique or stochastically choose the edges under a certain connectivity assumption. Network constructors \cite{Michail2016} can also be simulated by NS by simply allowing the network $G^{(t)}$ to be the network that evolves to the target network (e.g., a path graph) while the interaction graph contains a stochastically chosen edge (or many as a generalization) for which the local rule will apply. It is a simple exercise to trasfer the local rules in \cite{Michail2016,Michail2017NetworkCA} in our setting.

\section{A Simple Example of Emergence} \label{sec:min}

In this section we define the energy as $\mathcal{E}(u,v)=\min\{d_{G^{(t)}}(u),d_{G^{(t)}}(v)\}$, that is we choose the minimum degree of its two nodes as the new energy of the edge. In the following, we write $C^{(t)}=K_n$ to denote that the interaction graph is a clique. In this sense, the algorithm is always applied on all edges $(\forall t: C^{(t)}=K_n)$. To investigate convergence we take cases for the nodes in $G^{(t)}$ for arbitrary $t$. We prove that with this definition of energy the process converges in $O(n)$ steps, and provide a matching lower bound on the number of steps needed. In addition, we show that when the algorithm stabilizes the network has computed in effect the $\alpha$-core \cite{DBLP:conf/gd/BatageljMZ99} (implicitly this notion was defined in \cite{SzekeresW}) of the network $G^{(0)}$, which can be seen as emergent behavior of this simple algorithm.

\begin{theorem}
The process always converges in $O(n)$ steps, and there are cases where $\Omega(n)$ steps are required.
\end{theorem}
\begin{proof}
For all nodes $v\in V$, for any $t$, such that $d(v)<\alpha$ it holds that they become isolated (their degree is zero) in the next step. This is because $min\{d(u),d(v)\}\leq d(v)<\alpha$, for any $d(u)$. As soon as a node becomes isolated, it will be isolated forever since again $d(v)=0<\alpha$.

There can be $O(n)$ consecutive rounds where at least one node becomes isolated; after that we end up with a graph $G^{(t)}$ for which there are two sets of nodes: the set $V^{(t)}_0$ which contains all nodes with degree $0$ as well as the set $V^{(t)}_{\alpha}$ that contains all nodes with degree at least $\alpha$.

The question is how many rounds are needed for $V^{(t)}_{\alpha}$ to settle down. No edge with both endpoints in $V^{(t)}_{\alpha}$ will cease to exist at time $t+1$ since this would mean that some node in $V^{(t)}_{\alpha}$ has degree less than $\alpha$ which is a contradiction. Thus, the degrees do not decrease.

We notice that all nodes with degree at least $\beta$ will form a clique at time $t+1$. On the other hand, nodes of $V^{(t)}_{\alpha}$ with degree less that $\beta$ will not form any new edge, effectively having the same neighbors at time $t+1$. The exact same reasoning gives $G^{(t+1)} = G^{(t+2)}$ and thus we have convergence in at most 1 step, due to Lemma~\ref{lem:end_cond_general}.

The above discussion gives an upper bound of $O(n)$ time steps for convergence. Let us give a simple matching lower bound. If $G^{(0)}$ is a path, $\alpha =2$ and $\beta >2$, then the graph converges in $\lfloor \frac{n}{2} \rfloor$ time steps, since no edge will ever be created, and at every time step, the two ends of the remaining path become isolated.
\end{proof}

It is interesting to notice the similarity of our process, and the process of acquiring the $\alpha-core$ (or complementary the $(\alpha-1)-crust$) of a simple undirected graph (implicit in a lemma by \cite{SzekeresW}, however we use the more convenient statement by \cite{DBLP:conf/gd/BatageljMZ99}).

\begin{definition}
The $\alpha$-core $H$ of a graph $G$ is the unique maximal subgraph of $G$ such that for the degree of every node $u\in H$ it holds that $deg_H(u)\geq \alpha$. All nodes not in $H$ form the $\alpha-1$ crust of graph $G$.
\end{definition}

The $\alpha$-core is a notion that plays an important role in studying the clustering structure of social networks. Batagelj et al. \cite{DBLP:conf/gd/BatageljMZ99} proved that the following process efficiently computes the $\alpha$-core of a graph:
\begin{lemma}
Given a graph $G$ and a number $\alpha$, one can compute $G$'s $\alpha$-core by repeatedly deleting all nodes whose degree is less than $\alpha$.
\end{lemma}

What makes our process and the process of detecting the $\alpha$-core different is the fact that new edges can emerge in our process. However, we can disallow this by setting $\beta\geq n$.

\begin{lemma}
When $\mathcal{E}(u,v)=\min\{d_{G^{(t)}}(u),d_{G^{(t)}}(v)\}$ and $(\forall t: C^{(t)}=K_n)$ the dynamic process for any value of $\alpha$ and $\beta\geq n$ is essentially the same process with the one for detecting the $\alpha$-core. Furthermore, all isolated nodes in our process form the $(\alpha-1)$-crust of $G^{(0)}$, while the remaining graph forms the $\alpha$-core.
\end{lemma}
\begin{proof}
First of all, even if a node connects with any other node, its degree will be $n-1$. Thus, it holds that $min\{d(u),d(v)\}\leq n-1 < \beta$. This ensures that no edge will ever be formed by the dynamic process.

As far as existing edges are concerned, both processes delete edges where at least one endpoint $v$ has degree less than $\alpha$, since $min\{d(u),d(v)\}\leq d(v)<\alpha$, for any $d(u)$. By the same reasoning, these nodes (which, by definition, belong to the $(\alpha-1)$-crust) remain isolated forever. Furthermore, edges with both endpoints having degree at least $\alpha$ will be preserved as the minimum of their degrees will still be at least $\alpha$.
\end{proof}

\section{Convergence} \label{sec:convergence}

In this section we study the convergence of the algorithm $\mathcal{A}$. First, in \ref{ssec:degree} we study the case of $\alpha=\beta$ and prove convergence as well as provide bounds on the convergence speed. In \ref{ssec:Local_Rules} we prove convergence for a more general energy function and for the case where $\alpha<\beta$ but do not make any statements on the speed of convergence. In \ref{ssec:common_neighborhood} we prove that when the energy function takes into account nodes at distance $2$ then we cannot always guarantee convergence and we use this result to disprove a conjecture of Zhang et al. \cite{DBLP:conf/kdd/ZhangWWZ09} while also discussing how convergence depends on the choice of the interaction graph $C^{(t)}$.

\subsection{Symmetric Non-Decreasing Function on the Degrees} \label{ssec:degree}

We study the case where $\alpha=\beta$, as we are not only able to prove convergence of the process, but we also prove an upper bound on the number of steps needed for convergence.

We define the energy of an edge $(u,v)$ to be $\mathcal{E}(u,v)=f(d_{G^{(t)}}(u),d_{G^{(t)}}(v))$, where $f$ is a $\textit{proper}$ (symmetric and non-decreasing in both variables) function. The rule is always applied on all edges $(\forall t: C^{(t)}=K_n)$.

For the graph $G^{(t)}$, let $R^{(t)}(u,v)$ be an equivalence relation defined on the set of nodes $V$ for time $t$, such that $(u,v)\in R^{(t)}$ if and only if $d_{G^{(t)}}(u)=d_{G^{(t)}}(v)$. The equivalence class $R^{(t)}_i$ corresponds to all nodes with degree $d(R^{(t)}_i)$, where $i$ is the rank of the degree in decreasing order. This means that the equivalence class $R^{(t)}_1$ contains all nodes with maximum degree in $G^{(t)}$. Apparently, the maximum number of equivalence classes is $n=|V|$, since the degree can be in the range $[0,n-1]$. Let $|G^{(t)}|$ be the number of equivalence classes in graph $G^{(t)}$. 

Before moving to the proof, we give certain properties of the dynamic process that hold for all $t\geq 1$, that is they hold after at least one round of the process (initialization). These properties will be used in the proof for convergence. First, we show that nodes have an implicit hierarchy with respect to degrees. 

\begin{prop}\label{prop:Monotonicity}
If $d_{G^{(t)}}(u) \geq d_{G^{(t)}}(w)$, then $d_{G^{(t+1)}}(u) \geq d_{G^{(t+1)}}(w)$, for all $t\geq 1$.
\end{prop}
\begin{proof}
For any neighbor $v$ of $w$ in $G^{(t+1)}$ it holds that $\mathcal{E}^{(t)}(v,w) \geq \beta$. Then it also holds that $\mathcal{E}^{(t)}(v,u) \geq \beta$, since $f$ is non-decreasing, which means $v$ is also a neighbor of $u$ in $G^{(t+1)}$. 
\end{proof}

\noindent Nodes that have the same degree at time $t$, share the same neighbors at time $t+1$. 

\begin{prop}\label{prop:EqualDegree}
If $d_{G^{(t)}}(u) = d_{G^{(t)}}(w)$, then $N_{G^{(t+1)}}(u) = N_{G^{(t+1)}}(w)$.
\end{prop}
\begin{proof}
As in the proof of Property~\ref{prop:Monotonicity}, due to the equality of the degrees, it also holds that any neighbor $v$ of $u$ is a neighbor of $w$ and respectively any neighbor $v$ of $w$ is a neighbor of $u$.
\end{proof}

\noindent In the following, we discuss properties related to equivalence classes.

\begin{prop} \label{prop:ECNonIncreasing}
The number of equivalence classes in $G^{(t+1)}$ is less than or equal to the number of equivalence classes in $G^{(t)}$.
\end{prop}
\begin{proof}
By Property~\ref{prop:EqualDegree}, nodes that belong to the same equivalence class at time $t>0$ will always belong to the same equivalence class for all $t'>t$.
\end{proof}

\begin{prop}  \label{prop:GraphsSameEC}
If $G^{(t+1)}$ has the same number of equivalence classes as $G^{(t)}$, then $\forall i$, $|R^{(t)}_i|=|R^{(t+1)}_i|$, where $|R^{(t)}_i|$ is the number of nodes in the equivalence class $R^{(t)}_i$.
\end{prop}
\begin{proof}
Suppose that the above doesn't hold. Then, there is some $i$ for which $|R^{(t)}_i| \neq |R^{(t+1)}_i|$. This means that there must be two nodes in some equivalence class $R^{(t)}_j$ that landed to different classes in $G^{(t+1)}$. However, Property~\ref{prop:EqualDegree} implies that this is impossible. 
\end{proof}

The following lemma describes how equivalence classes behave with respect to edge distribution.

\begin{lemma} \label{lem:MonotonicityEC}
If an arbitrary node $u$ in $R^{(t)}_i$ is connected with some node $w$ in $R^{(t)}_j$, then $u$ is connected with every node $x$ in every equivalence class $R^{(t)}_k$, such that $k\leq j$ and $t>0$.
\end{lemma}
\begin{proof}
Due to Property~\ref{prop:Monotonicity}, for all nodes $x\in R^{(t)}_k$ it holds that $d_{G^{(t)}}(x) \geq d_{G^{(t)}}(w)$ and so they are also neighbors of $u$.
\end{proof}

We prove by induction that this process always converges in at most $2|G^{(0)}|$ steps. To begin with, it is obvious that the clique $\mathcal{K}_n$ as well as the null graph $\overline{\mathcal{K}_n}$ both converge in at most one step, for any value of $\beta$. The following renormalization lemma describes how the number of equivalence classes is reduced and is crucial to the induction proof. 

\begin{lemma}	\label{lem:Renormalization}
If $d\left(R^{(t)}_1\right) = n-1$ for every $t\geq c$, for some $c\in \mathbb{N}\cup \{0\}$, and the subgraph $G^{(c)} \setminus R^{(c)}_1$ converges for any value of $\beta$ and proper function $f$, then $G^{(c)}$ converges as well. Similarly, if $d\left(R^{(t)}_{|G^{(t)}|}\right) = 0$ for every $t\geq c$, for some $c\in \mathbb{N}\cup \{0\}$, and the subgraph $G^{(c)} \setminus R^{(c)}_{|G^{(c)}|}$ converges for any value of $\beta$ and proper function $f$, then $G^{(c)}$ converges as well. The time it takes for $G^{(c)}$ to converge is the same as the time it takes for the induced subgraph to converge, for both cases.
\end{lemma}

\begin{proof}
The main idea is that we consider two different sets of nodes: $R^{(c)}_1$ and $V\setminus R^{(c)}_1$. Due to our hypothesis, at all future time-steps the edges between these two groups, and the edges with both endpoints in $R^{(c)}_1$ are fixed. Concerning the edges with both endpoints in $V\setminus R^{(c)}_1$, we can almost study this subgraph independently. That's because the effect of $R^{(c)}_1$ on $V\setminus R^{(c)}_1$ is completely predictable: it always increases the degree of all nodes by the exact same amount. The same reasoning applies for $R^{(c)}_{|G^{(c)}|}$.

More formally, by Property~\ref{prop:Monotonicity}, for all $t\geq c$ it holds that $R^{(t)}_1 \subseteq R^{(t+1)}_1$. This means that the nodes in $R^{(c)}_1$ are always connected to every node after time $c$. As a result, for all $u\in V\setminus R^{(c)}_1$ it holds that their degree in the induced subgraph $G^{(t)}\setminus R^{(c)}_1$ is $d_{G^{(t)}}(u)-|R^{(c)}_1|$. 
Thus, the decision for the existence of an edge $(u,v)$, where $u,v\in G^{(t)}\setminus R^{(c)}_1$ is the following:
\[\mathcal{E}^{(t)}(u,v)=f(d_{G^{(t)}\setminus R^{(c)}_1}(u)+|R^{(c)}_1|,d_{G^{(t)}\setminus R^{(c)}_1}(v)+|R^{(c)}_1|)\geq \beta\]
which can be written as: 
\[\mathcal{E}^{(t)}(u,v)=g(d_{G^{(t)}\setminus R^{(c)}_1}(u),d_{G^{(t)}\setminus R^{(c)}_1}(v))\geq \beta\]
where
\[g(x,y)=f(x+|R^{(c)}_1|,y+|R^{(c)}_1|)\]

Clearly, $g$ is a proper function assuming that $f$ is a proper function. Thus, the choice of whether the edge exists between $u$ and $v$ is equivalent between $G^{(t)}$ and $G^{(t)}\setminus R^{(c)}_1$ by  appropriately changing $f$ to $g$. But due to our hypothesis $G^{(c)}\setminus R^{(c)}_1$ converges, and thus $G^{(c)}$ also converges in the same number of steps. Note that we need not compute $g$ since this is only an analytical construction; the dynamic process continues as defined.
The proof of the second part of the lemma is similar in idea but much simpler since function $f$ does not change due to the fact that the removed nodes have degree $0$.
\end{proof}

The following theorem establishes that the dynamic process converges in linear time. 

\begin{theorem} \label{thm:beta_degree}
When $\alpha=\beta$, $\mathcal{E}(u,v)=f(d_{G^{(t)}}(u),d_{G^{(t)}}(v))$, and $(\forall t: C^{(t)}=K_n)$, the dynamic process on an undirected simple graph $G$ converges in at most $2|G^{(0)}|$ steps.
\end{theorem}
\begin{proof}
We use induction on the number of equivalence classes. For the base case, the graph $G^{(0)}$ has only one equivalence class $R^{(0)}_1$ (the graph is regular). There are two cases: either $f(d(R^{(0)}_1),d(R^{(0)}_1)) \geq \beta$ and all edges are created $(G^{(1)} = \mathcal{K}_n)$, or $f(d(R^{(0)}_1),d(R^{(0)}_1))<\beta$ and no edge is created $(G^{(1)} = \overline{\mathcal{K}_n})$. Either way, $G^{(1)}$ converges in at most $1$ step, and thus $G^{(0)}$ converges in at most $2$ steps.

Suppose the theorem holds for $i-1$ equivalence classes and let $G^{(t)}$ be a graph with $|G^{(t)}| = i > 1$. If the number of equivalence classes decreases within the first two steps, then the process converges in at most $2i$ steps, by the inductive hypothesis. Thus, we only look at the case where the number of equivalence classes remains the same. Our main idea is to take advantage of the following: if the node with maximum degree connects with the node with minimum degree, then it also connects with every other node in the graph, and its degree will be $n-1$. Else, the minimum degree node will become isolated (degree $0$).
We discern four different cases in total, concerning the relation of $d(R^{(t)}_1)$, $d(R^{(t)}_i)$ and $\beta$.
  \begin{enumerate}
  \item $f(d(R^{(t)}_1),d(R^{(t)}_1)) < \beta$ \label{cs:zero}
  \item $f(d(R^{(t)}_1),d(R^{(t)}_i)) < \beta \leq f(d(R^{(t)}_1),d(R^{(t)}_1))$ \label{cs:decrease}
  \item $f(d(R^{(t)}_i),d(R^{(t)}_i)) < \beta \leq f(d(R^{(t)}_1),d(R^{(t)}_i))$ \label{cs:increase}
  \item $\beta \leq f(d(R^{(t)}_i),d(R^{(t)}_i))$ \label{cs:clique}
  \end{enumerate}

The proof of convergence is based on the fact that Cases~(\ref{cs:decrease}) and (\ref{cs:increase}) can only interchange once. This is based on the fact that the degree of an equivalence class will, at some time, be either $n-1$ or $0$ and thus by using Lemma~\ref{lem:Renormalization} we reduce the number of equivalence classes and finally the inductive hypothesis proves the theorem. To begin with, Cases~(\ref{cs:zero}) and (\ref{cs:clique}) would result in $\overline{\mathcal{K}_n}$ and $\mathcal{K}_n$ respectively, and thus $G^{(t)}$ would converge in at most $2$ steps.

Case~(\ref{cs:decrease}) results in $G^{(t+1)}$ such that $d(R^{(t+1)}_{|G^{(t+1)}|}) = 0$. If $|G^{(t+1)}|<|G^{(t)}|=i$ then $G^{(t)}$ converges in at most $2|G^{(t+1)}|+1< 2|G^{(t)}|$ steps. Else it holds that:
\[f(d(R^{(t)}_1),d(R^{(t)}_{i-1}))\geq \beta\]
because otherwise $R^{(t)}_i$ and $R^{(t)}_{i-1}$ would be joined in a single equivalence class $G^{(t+1)}$, effectively reducing the number of equivalence classes. Thus, $d(R^{(t+1)}_1) = n-|R^{(t)}_i|-1$, due to Lemma~\ref{lem:MonotonicityEC}, Property~\ref{prop:GraphsSameEC} and the fact that $d(R^{(t+1)}_i) = 0$. In the case where
\[f(n-|R^{(t)}_i|-1,0)< \beta\]
we always get Cases~(\ref{cs:zero}) or (\ref{cs:decrease}) because, inductively, the minimum degree will always be $0$, while the maximum degree will be at most $n-|R^{(t)}_i|-1$. In this case, the theorem is proved due to Lemma~(\ref{lem:Renormalization}) and the inductive hypothesis.

On the other hand, if
\[f(n-|R^{(t)}_i|-1,0)\geq \beta\]
then we always get Cases~(\ref{cs:increase}) or (\ref{cs:clique}) since, the maximum degree will always be $n-1$ as we prove below. In this case, the theorem is also proved due to Lemma~(\ref{lem:Renormalization}) and the inductive hypothesis.

The same reasoning works for Case~(\ref{cs:increase}), which results in $d(R^{(t+1)}_1) = n-1$. Like before, we assume $|G^{(t+1)}|=|G^{(t)}|$ (otherwise the equivalence classes are reduced) and so
\[f(d(R^{(t)}_2),d(R^{(t)}_{i}))< \beta\]
Thus, $d(R^{(t+1)}_i) = |R^{(t)}_1|$, due to Lemma~\ref{lem:MonotonicityEC} and Property~\ref{prop:GraphsSameEC}. In the case where
\[f(n-1,|R^{(t+1)}_i|)\geq \beta\]
then from this point on we always get Cases~(\ref{cs:increase}) or (\ref{cs:clique}) since, inductively, the maximum degree will always be $n-1$ and the minimum degree will always be at least $|R^{(t+1)}_i|$.
On the other hand, if
\[f(n-1,|R^{(t+1)}_i|)< \beta\]
then we always get Cases~(\ref{cs:zero}) or (\ref{cs:decrease}) because, inductively, the minimum degree will always be $0$. 

In all possible cases, after at most $2$ rounds it suffices to examine graphs with reduced number of equivalence classes. This proves the upper bound for the convergence.
\end{proof}

\subsection{Extending the Energy on the Degrees} \label{ssec:Local_Rules}

In this section we extend the update rule given in \ref{ssec:degree}. More specifically, we change the definition of energy, from $\mathcal{E}(u,v)=f(d_{G^{(t)}}(u),d_{G^{(t)}}(v))$ to $\mathcal{E}(u,v)=f(g_{G^{(t)}}(u),g_{G^{(t)}}(v))$, for a family of functions $g_{G}:\mathbb{R}^k\rightarrow \mathbb{R}, k\in\mathbb{N}$.

We call a function $g_G(u)$ \textit{degree-like} if it only depends on the neighborhood $N_G(u)$ of node $u$. This dependence is formally translated as follows: assuming that the neighborhood of node $u$ at time $t$ is $N_{G^{(t)}}(u)$, the neighborhood of node $v$ at time $t'$ is $N_{G^{(t')}}(v)$, and $N_{G^{(t)}}(u) \supseteq N_{G^{(t')}}(v)$, then $g_{G^{(t)}}(u) \geq g_{G^{(t')}}(v)$.  Notice that, generally, the values $t$ and $t'$ may differ. The reason we extend the notion of degree is so that $g$ can represent more interesting rules. For example, we are no longer obliged to handle all nodes in the same manner; nodes can be assigned an importance factor (e.g. a known centrality measure such as their betweenness centrality in $G^{(0)}$), and let $g(u)$ be the sum of these factors of nodes in the neighborhood of $u$.

Additionally, at any time and for any edge we can arbitrarily decide whether the rule will be applied. This means that $C^{(t)}$ can change for different values of $t$, with no restrictions posed. For example, allowing only the preservation of edges from time $t_0$ to time $t_0+1$ would be achieved by setting $C^{(t_0)}=E^{(t_0)}$, and applying the rules only on pairs of nodes whose distance is at most $2$ would be $C^{(t_0)}= \{(u,v)~s.t.~((u,v)\in E^{(t_0)})~or~(\exists~w~s.t.~((u,w),(w,v)) \in E^{(t_0)}\times E^{(t_0)})\}$. We also assume that the function $f$ is proper (symmetric and non-decreasing in both variables). It is easy to see that the update rule in \ref{ssec:degree} is a special case of the current update rule, where the function $g$ is the degree of the node, and $\forall t: C^{(t)}=K_n$.

Notice that the introduction of $C^{(t)}$ allows us to define local update rules. For example, $C^{(t)}$ could be defined in a way that allows an edge to be formed if and only if the previous distance between the two nodes is bounded by some constant $K$.

To show that any such process converges, we define the following:

\begin{definition}
A pair $(t,D)$ is said to be $|D|-Done$ if $t$ is a natural number, $D \subseteq V$ and it holds that the neighborhood of all nodes $u \in D$ doesn't change after time $t$. That is, $N_{G^{(t')}}(u) = N_{G^{(t)}}(u)$, for $t'\geq t$.
\end{definition}

Our convergence proof repeatedly detects $|D|-Done$ pairs with increasing $|D|$. When $D=V$, all neighborhoods do not change, and thus the process converges.

\begin{lemma}\label{lem:Increasing_D_Done}
If there exists a $|D|-Done$ pair $(t,D)$ at time step $t$, then $\exists t'>t$ such that at time step $t'$ there exists a $(|D|+1)-Done$ pair $(t',D')$.
\end{lemma}
\begin{proof}
Let $t_1\geq t$ be a time-step where some node $u \not\in D$ maximizes the function $g$ over all future time-steps and nodes not in $D$. More formally, we define $t_1\geq t$ as the time-step where there is some node $u \not\in D$ such that $g_{G^{(t_1)}}(u)\geq g_{G^{(t_1')}}(v)$, for all $t_1'\geq t_1$ and $v \in V\setminus D$. If there are many such choices, we arbitrarily pick one where the degree of $u$ is the highest. Let us note that, later in time (say at $t_1'>t_1$), it is entirely possible that $u$'s neighborhood shrinks and thus its $g$ value drops $(g_{G^{(t_1')}}(u) < g_{G^{(t_1)}}(u))$. 

It is guaranteed that $t_1$ exists, as there are finitely many graphs with $|V|$ nodes, and finitely many nodes. Thus, there are finitely many values of $g_{G}(u)$ to appear after time $t$.

Our core idea is that either $u$'s neighborhood stays the same in all subsequent time-steps (and thus $D$ is extended by $u$), or some edge is lost along the way. But if the other endpoint $w$ of the edge can't preserve an edge with $u$, which maximizes $g$, then it doesn't preserve any other edge. Inductively, it will never form any new edge, and thus $D$ can be extended by $w$.

More formally, if neighbors of $u$ in $G^{(t_1)}$ remain neighbors of $u$ in all subsequent time-steps, then, in future time-steps, its neighborhood can only grow from $N_{G^{(t_1)}}(u)$, or stay the same. But if its neighborhood grows, due to the properties of function $g$, its value will not drop and the degree of $u$ will increase. However, the way we picked $u$ doesn't allow this. We conclude that the neighborhood of $u$ doesn't change after time $t_1$, and thus we can extend $D$ by $\{u\}$, that is $(t_1,D \cup \{u\})$ is $(|D|+1)-Done$.

Else, let $t_2>t_1$ be the first time-step that a neighbor $w$ of $u$ in $G^{(t_2-1)}$ is not a neighbor of $u$ in $G^{(t_2)}$. It follows directly from the fact that $u$'s neighborhood stays the same until $t_2-1$ that $g_{G^{(t_1)}}(u)=g_{G^{(t_2-1)}}(u)$. 
Then $w$ has no neighbor $v\in V\setminus D$ in $G^{(t_2)}$, as it holds that $\alpha> f(g_{G^{(t_2-1)}}(u),g_{G^{(t_2-1)}}(w)) = f(g_{G^{(t_1)}}(u),g_{G^{(t_2-1)}}(w)) \geq f(g_{G^{(t_2-1)}}(v),g_{G^{(t_2-1)}}(w))$. The latter inequality follows from the way we picked $t_1$ and $u$. Of course, due to the definition of $D$, no new edge is formed between $w$ and a node in $D$. Thus, the neighborhood of $w$ shrinks, and due to $g$'s properties $g_{G^{(t_2-1)}}(w) \geq g_{G^{(t_2)}}(w)$.

We argue that the neighborhood of $w$ at all subsequent time steps will stay the same, that is $N_{G^{(t_2)}}(w)=N_{G^{(t_2')}}(w)$, $t_2'\geq t_2$. We prove this inductively. It trivially holds for $t_2'=t_2$. Supposing it holds for some $t_2'$, we prove that it also holds for $t_2'+1$. If it doesn't, then $w$ forms an edge with some node $v \in V \setminus D$, due to the definition of $D$. But we know that $\beta\geq \alpha > f(g_{G^{(t_2-1)}}(u),g_{G^{(t_2-1)}}(w)) = f(g_{G^{(t_1)}}(u),g_{G^{(t_2-1)}}(w)) \geq f(g_{G^{(t_2')}}(v),g_{G^{(t_2')}}(w))$ due to $f$ being non-decreasing. We conclude that the neighborhood of $w$ doesn't change after time $t_2$, and thus we can extend $D$ by $\{w\}$, that is $(t_2,D \cup \{w\})$ is $(|D|+1)-Done$.
\end{proof}

\begin{theorem} \label{thm:conv_only}
When $\mathcal{E}(u,v)=f(g_{G^{(t)}}(u),g_{G^{(t)}}(v))$, the dynamic process on an undirected simple graph $G$ converges for any $\alpha, \beta$, proper function $f$, degree-like function $g$ and sequence of interaction graphs $C^{(t)}$.
\end{theorem}
\begin{proof}
It trivially holds that $(0,\emptyset)$ is $0-Done$. By applying Lemma~\ref{lem:Increasing_D_Done} once, we increase the size of $D$ by $1$. Thus, by applying it $|V|$ times, we end up with a $|V|-Done$ pair $(t,V)$. Since all neighborhoods stay the same for all future steps, $G^{(t')}=G^{(t)}$ for all $t'\geq t$.
\end{proof}

\subsection{Moving Beyond Degree} \label{ssec:common_neighborhood}

We define the energy of an edge $(u,v)$ to be $\mathcal{E}(u,v)=|N_{G^{(t)}}(u)\cap N_{G^{(t)}}(v)|$. The rule is only applied on pairs of nodes whose distance is at most $2$ (they are either directly connected or they share a common neighbor). We prove that the process may not converge. To prove this, we provide a certain family of graphs, such that when $G^{(0)}$ is any member of this family and the parameters are $\alpha=\beta=2$, the process doesn't converge.
Furthermore, this family has the property that for any positive number $c$, there exists a member of it such that the cycle size of the process with $\alpha=\beta=2$ is at least $c$. Finally, we provide examples where consecutive graphs do not share any edges at all; thus, even if we stop when consecutive graphs are "close enough", and not necessarily the same, the process still doesn't stabilize. We even give an example where we infinitely swap between a graph and its complement, which is the farthest we could get from "close enough".

\begin{definition} \label{def:GS}
Let $[S]$ denote the set $\{0,1,...,S-1\}$. For each $S$, we define the graph $G_S=(V_S,E_S)$ to be the undirected graph with $S^2$ nodes, where each node is identified by a pair of integers (that is $V_S=[S]\times[S]$) and for each node $(x,y)$ its neighbors are the nodes $(x,y\pm1 \mod S)$ and $(x\pm1 \mod S,y)$.
\end{definition}

One can think of $G_S$ as a $2$-dimensional $S\times S$ grid, such that each point is connected with its four immediate neighbors (up/down/left/
right). These connections are considered modulo $S$, so for a leftmost point at height $y$, its left neighbor is actually the rightmost point at
height $y$, and similarly for the other $3$ directions.

Notice that applying the dynamic process to $G^{(0)}=G_3$, we infinitely swap between $G_3$ and its complement.

Except for the case where $S=1$, which trivially converges, we prove that the process with $G^{(0)}=G_S$ and parameters $\alpha=\beta=2$ doesn't converge, for all other odd values of $S$. To prove this, we need the following lemma.

\begin{lemma}\label{lem:Distinct_Mod}
Let $i,j$ be natural numbers (including zero), and $S>1$ be an odd number. Then it holds that $i \not\equiv i+2^j \mod S$.
\end{lemma}
\begin{proof}
Suppose that $i \equiv i+2^j \mod S$. Then $i+2^j=i+kS \implies 2^j=kS$ for some integer $k$. But this is impossible, since $2^j$ only contains the number $2$ in its (unique) prime factorization, while $kS$ contains at least one odd prime, due to $S$ being odd.
\end{proof}

Using the above, we are ready to describe all $G^{(t)}$ in the process.

\begin{figure}
\begin{center}
\includegraphics[scale=0.5]{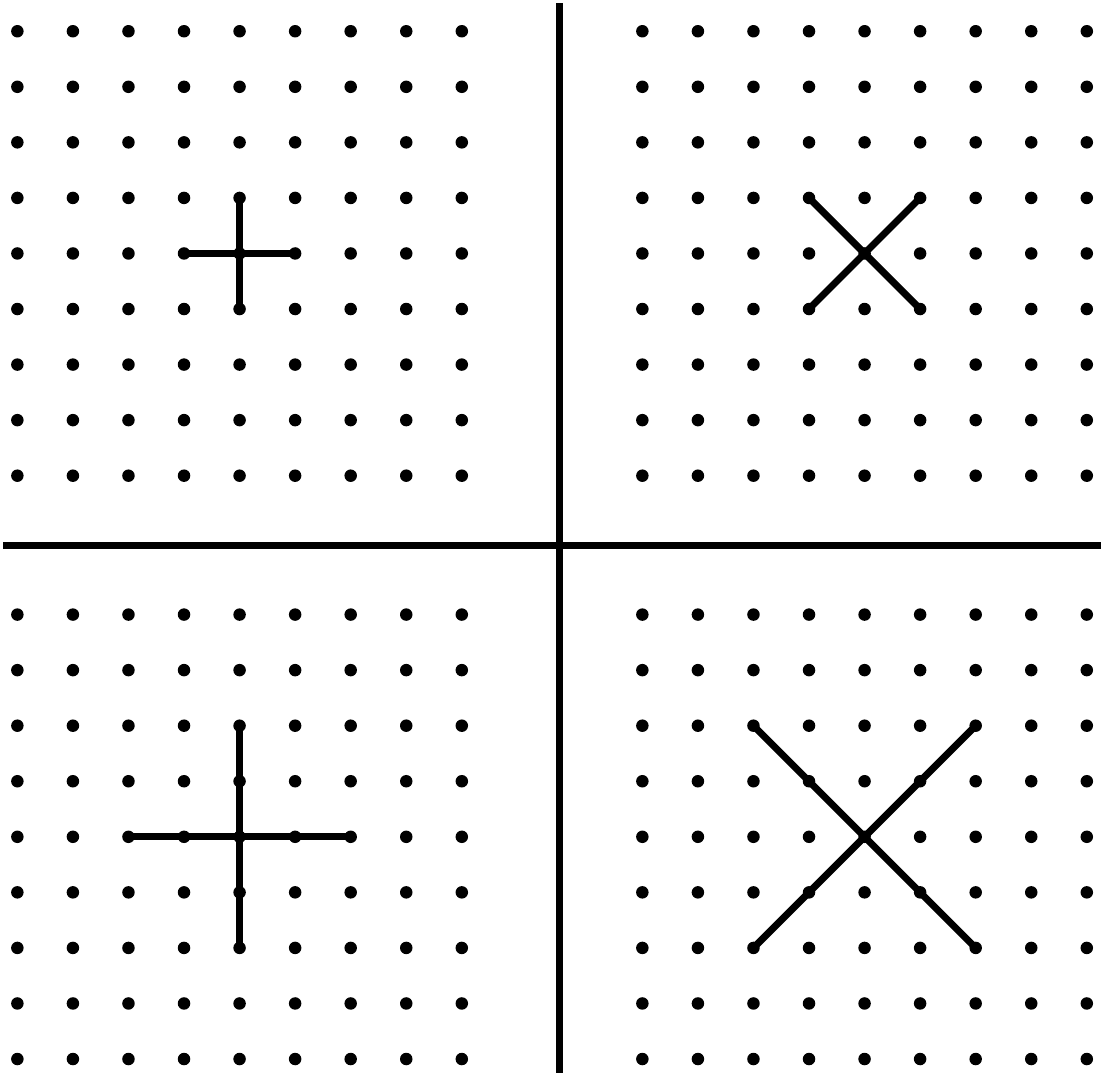}
\end{center}
\caption{We only show the $4$ neighbors of the middle node, as others follow the same pattern. The top-left figure shows the neighbors at $t=0$, the top-right at $t=1$, the bottom-left at $t=2$ and the bottom right at $t=3$.}
\label{fig:counter_example_divergence}
\end{figure}

\begin{lemma}\label{lem:Counterexample_Process}
Let $S\geq 3$ be an odd integer, and the process have parameters $\alpha=\beta=2$ and $G^{(0)}=G_S$. Then any node $(x,y)$ in $G^{(t)}$ has the following 4 neighbors:
\begin{itemize}
	\item $(x \pm 2^l \mod S,y)$ and $(x,y \pm 2^l \mod S)$, if $t=2l$ is even.
	\item $(x \pm 2^l \mod S, y \pm 2^l \mod S)$, if $t=2l+1$ is odd.
\end{itemize}
\end{lemma}
\begin{proof}
A visual representation of $G^{(t)}$ is given in Figure~\ref{fig:counter_example_divergence}. Intuitively, the same way we described $G_S$ as a grid where each point is connected with its $4$ immediate neighbors, $G^{(t)}$ can be thought of as a grid where each point is again connected with $4$ other points. If $t=2l$ is even, then we find these neighbors by picking one of the $4$ directions (up, down, left, right) and walking $2^l$ steps. If $t=2l+1$ is odd, then we pick one of the $4$ diagonal directions (up-left, up-right, down-right, down-left) and walk $2^l$ steps. Again, this process is done modulo $S$.

We use induction to prove our lemma.  The lemma holds for $t=0$ due to the definition of $G^{(0)}$. Suppose it holds for $t$, we show that it also holds for $t+1$. If $t$ is even, $t=2l$, then due to our inductive hypothesis, the neighbors of $(x,y)$ are $(x \pm 2^l \mod S,y)$ and $(x,y \pm 2^l \mod S)$. These $4$ nodes are all distinct with each other and distinct from $(x,y)$. We show this for just one pair, namely $(x+2^l \mod S,y)$ and $(x-2^l \mod S,y)$, as all others follow the same reasoning. Suppose they coincided; then $x+2^l \equiv x-2^l \mod S \implies x+2^{l+1} \equiv x \mod S$, which is not allowed by Lemma~\ref{lem:Distinct_Mod}.

To find nodes sharing common neighbors with $(x,y)$, it suffices to check at neighbors of $(x,y)$'s neighbors. Each of these $4$ nodes only has $4$ neighbors. There are $16$ such nodes, but since $(x,y)$ obviously appears $4$ times, only $12$ nodes are of interest. We see that the nodes $(x \pm 2^l \mod S, y \pm 2^l \mod S)$ appear in the neighborhood of $\beta=2$ of $(x,y)$'s neighbors (and thus form an edge with $(x,y)$ at time $t+1$, since they share $\beta$ common neighbors with it). The nodes $(x \pm 2^{l+1} \mod S, y)$ and $(x, y \pm 2^{l+1} \mod S)$ appear in the neighborhood of only $1$ of $(x,y)$'s neighbors, and thus do not form an edge with $(x,y)$ at time $t+1$. To complete the proof, we use the technique of the previous paragraph to show that all aforementioned nodes are distinct.

The case where $t=2l+1$ is completely analogous.
\end{proof}

Now that we have a description of all $G^{(t)}$, it is easy to see that the process doesn't converge.

\begin{lemma} \label{lem:Cycle_Size}
Let $S\geq 3$ be an odd integer, and the process have parameters $\alpha=\beta=2$ and $G^{(0)}=G_S$. Then the process doesn't converge, and the cycle size is $2k$, where $k>0$ is the smallest integer such that $2^k \equiv \pm 1 \mod S$.
\end{lemma}
\begin{proof}
First of all we prove that such a number $k$ exists. Due to the pigeonhole principle, there is some pair of integers $i,j$, where $0\leq i<j\leq S$ such that $2^i \equiv 2^j \mod S$, which means that $2^i \equiv 2^i\times 2^{j-i} \mod S$. It follows that $2^i(1-2^{j-i}) = z_1S$ for some integer $z_1$, and since $2^i$ and $S$ do not share any common prime factors (due to $S$ being odd), then $z_1=2^iz_2$ for some integer $z_2$. Thus $1-2^{j-i}=z_2N \implies 2^{j-i} \equiv 1 \mod S$, which proves our point.

We notice that $G^{(0)}$ is different from every $G^{(t)}$ for odd $t=2l+1$. To show this, we note that $(0,0)$ and $(0,1)$ are neighbors at $G^{(0)}$, but all neighbors of $(0,0)$ at $G^{(t)}$ are distinct from $(0,1)$. That's because if any of $(\pm 2^l \mod S, \pm 2^l \mod S)$ coincides with $(0,1)$, then it holds that $2^l \equiv 0 \mod S$, which is not allowed by Lemma~\ref{lem:Distinct_Mod}.

It is also straightforward to verify that if $2^t \equiv \pm 1 \mod S$, then $G^{(0)} = G^{(2t)}$, using Lemma~\ref{lem:Counterexample_Process}. On the other hand, if $2^t \not\equiv \pm 1 \mod S$, then $G^{(0)} \neq G^{(2t)}$, since the edge connecting $(0,0)$ and $(0,1)$ in $G^{(0)}$ doesn't correspond to any edge in $G^{(2t)}$. The latter follows from the fact that $(0,0)$, in $G^{(2t)}$ is connected to $(\pm 2^t \mod S, 0)$, which both differ from $(0,1)$, and to $(0,\pm 2^t \mod S)$, which also differ because we assumed $2^t \not\equiv \pm 1 \mod S$.
\end{proof}

The above discussion naturally leads us to our main theorem concerning the convergence of this process.

\begin{theorem}\label{thm:Common_Neibs_Cycle}
Let $c$ be any natural number. When $\mathcal{E}(u,v)=|N_{G^{(t)}}(u)\cap N_{G^{(t)}}(v)|$ and the rule is only applied on pairs of nodes whose distance is at most $2$, it is feasible to find a value $S(c)$ such that the process with parameters $\alpha=\beta=2$ and $G^{(0)} = G_{S(c)}$ has a cycle size of at least $c$.
\end{theorem}
\begin{proof}
Picking $S(c)=2^{\lceil \frac{c}{2} \rceil}+1$, we have that the process has a cycle size of $2k$, where $k>0$ is the smallest integer such that $2^k \equiv \pm 1 \mod S(c)$, due to Lemma~\ref{lem:Cycle_Size}. For $k=\lceil \frac{c}{2} \rceil$ we have that $2^k \equiv -1 \mod S(c)$. All values $t$, where $0<t<k$ have $1<2^{t}<S(c)-1$, and thus $2^{t} \not\equiv \pm 1 \mod S(c)$. The cycle size is therefore $2\lceil \frac{c}{2} \rceil \geq c$.
\end{proof}

\subsubsection{Dependence of Convergence on \texorpdfstring{$C^{(t)}$}{C(t)}} \label{sssec:dependence}

In Lemma~\ref{lem:end_cond_general}, we assumed that $C^{(t)}=K_n$ to prove that the stabilization condition is $G^{(t'-1)}=G^{(t')}$. At this point we show that there are cases where if it holds that $\forall t: C^{(t)}=K_n$ then convergence is guaranteed while other choices of $C^{(t)}$ for the same initial graph lead to infinite loops. To demonstrate this, we use the machinery developed to prove Theorem~\ref{thm:Common_Neibs_Cycle}. Consider the two following instances of the problem. Both instances have $\mathcal{E}^{(t)}(u,v)=|N_{G^{(t)}}(u) \cap N_{G^{(t)}}(v)|$ if $d_{G^{(t)}}(u)=d_{G^{(t)}}(v)=4$, and $\mathcal{E}^{(t)}(u,v)=0$ otherwise, $G^{(0)}=\mathcal{K}_9$ (clique with $9$ nodes) and $\alpha=\beta=2$. On the first instance, it holds that $\forall t: C^{(t)}=K_n$. On the second one, it holds that $C^{(0)}$ contains the edges of $G_3$ (which is defined in Definition~\ref{def:GS}), and for $t>0: C^{(t)}=K_n$.

Since $d_{G^{(0)}}(u)=d_{G^{(0)}}(v)=8\neq 4$, the energy of all pairs of nodes is zero. Thus, on the first instance we get $G^{(1)}=\overline{\mathcal{K}_9}$ (the null graph with $9$ nodes). But since all nodes are isolated on $G^{(1)}$, we get that $d_{G^{(1)}}(u)=d_{G^{(1)}}(v)=0\neq 4$ for all pairs of nodes $(u,v)$, which implies convergence, due to the fact that $G^{(1)}=G^{(2)}=\overline{\mathcal{K}_9}$. On the second instance, the same reasoning for $G^{(0)}$ gives $G^{(1)}=G_3$, due to $C^{(0)}$, which preserves edges. But then it is trivial to use Lemma~\ref{lem:Cycle_Size} to prove that this process doesn't converge.

\subsubsection{Disproving a Convergence Conjecture} \label{sssec:disproveconj}

Zhang et al. defined, in \cite{DBLP:conf/kdd/ZhangWWZ09}, the energy of an edge to be $\mathcal{E}(u,v)=| N_{G^{(t)}}(u) \cap N_{G^{(t)}}(v) | + | E^{(t)}(u,v) | + | E(G[N_{G^{(t)}}(u) \cap N_{G^{(t)}}(v)]) |$. This is an extension of our model described above, where we also add the number of edges between common neighbors of $u$ and $v$, denoted by $\left| E(G[N_{G^{(t)}}(u) \cap N_{G^{(t)}}(v)]) \right|$, and the binary term $\left| E^{(t)}(u,v) \right|$, which is the number of edges between $u$ and $v$. The rule is only applied on pairs of nodes whose distance is at most $2$.

Zhang et al. proposed the above process as an enhancement of the contrast between communities so that a community detection algorithm can discover them more easily. They conjectured that this process always converges. However, we disprove their conjecture by providing a counterexample. As in \ref{ssec:common_neighborhood}, we provide examples where consecutive graphs do not share any edges at all; thus, the above holds even if we stop when consecutive graphs are "close enough", and not necessarily the same.

Our proof is heavily based on the counterexamples given above. We first prove that these  counterexamples do not contain any triangles for certain values of $S$.

\begin{lemma}\label{lem:No_Triangles}
Let $S\geq 5$ be any odd integer not divisible by $3$, $\mathcal{E}(u,v)=|N_{G^{(t)}}(u)\cap N_{G^{(t)}}(v)|$, $\alpha=\beta=2$ and $G^{(0)}=G_S$. Then $G^{(t)}$ doesn't contain any triangle, for any $t$.
\end{lemma}
\begin{proof}
For simplicity, suppose that $t$ is even, $t=2l$, as the other case is analogous. Suppose that a triangle exists; then, due to the symmetry of $G^{(t)}$ we can assume that $2$ of the $3$ nodes are the neighboring nodes $(i,j)$ and $(i,j+2^l \mod S)$ (Lemma~\ref{lem:Counterexample_Process}). Since the third node is a neighbor of $(i,j+2^l \mod S)$, then it is either one of $(i \pm 2^l \mod S,j + 2^l \mod S)$ or $(i,j+2^{l+1} \mod S)$. But if $(i \pm 2^l \mod S,j + 2^l \mod S)$ were neighbors with $(i,j)$, then either $i \equiv i+2^l \mod S$ or $j=j+2^l \mod S$, which doesn't hold, due to Lemma~\ref{lem:Distinct_Mod}.

Thus, the third node must be $(i,j+2^{l+1} \mod S)$. For $(i,j)$ to be neighbors with $(i,j+2^{l+1} \mod S)$, it holds that $2^{l+1} \equiv -2^l \mod S \implies 3\times 2^l = z_1S$, where $z_1$ is integer. However, $S$ is odd and not divisible by $3$. Thus $z_1=3\times 2^l\times z_2$ for some integer $z_2$, which means $z_2S=1 \implies S=1$. But this is a contradiction, as $S\geq 5$.
\end{proof}

We are now ready to prove that for certain values of $S$, the counterexamples of \ref{ssec:common_neighborhood} produce the same sequence of graphs for both processes, and thus the current process doesn't always converge.

\begin{lemma}\label{lem:Common_Counterexamples}
Let $S\geq 5$ be any odd integer not divisible by $3$. Then, when $\alpha=\beta=2$ and $G^{(0)}=G_S$, both the process with $\mathcal{E}(u,v)=|N_{G^{(t)}}(u)\cap N_{G^{(t)}}(v)|$ and the process with $\mathcal{E}(u,v)=| N_{G^{(t)}}(u) \cap N_{G^{(t)}}(v) | + | E^{(t)}(u,v) | + | E(G[N_{G^{(t)}}(u) \cap N_{G^{(t)}}(v)]) |$ produce the same sequence of graphs.
\end{lemma}
\begin{proof}
Let $G'^{(t)}$ be the graphs of the process with $\mathcal{E}(u,v)=|N_{G^{(t)}}(u)\cap N_{G^{(t)}}(v)|$ and $G^{(t)}$ be the graphs of the other process. It trivially holds that $G'^{(0)}=G^{(0)}$. Suppose that it holds that $G'^{(t)}=G^{(t)}$. We prove that $G'^{(t+1)}=G^{(t+1)}$. First of all, the number of edges between common neighbors at time $t$ is the same, since both graphs are equal. That is $\left| E(G[N_{G^{(t)}}(u) \cap N_{G^{(t)}}(v)]) \right| = \left| E(G[N_{G'^{(t)}}(u) \cap N_{G'^{(t)}}(v)]) \right|$, which, due to Lemma~\ref{lem:No_Triangles}, is equal to $0$. If there is no edge between $u$ and $v$ in $G^{(t)}$, the energy of $(u,v)$ is the same in both processes, while, if an edge exists, the energy in $G^{(t)}$ is equal to the energy in $G'^{(t)}$ plus $1$, due to the term $\left| E^{(t)}(u,v) \right|$. But, even though it is incremented by $1$, it still holds that $\mathcal{E}^{(t)}(u,v) < \alpha = 2$. That is because $(u,v)$ do not have any common neighbor, for, if they had, these three nodes would form a triangle at time $t$ and this would contradict Lemma~\ref{lem:No_Triangles}. Thus, $G'^{(t+1)}=G^{(t+1)}$.
\end{proof}

We are now ready to prove our main theorem for this process.

\begin{theorem}
Let $c>4$ be any natural number. When $\mathcal{E}(u,v)=| N_{G^{(t)}}(u) \cap N_{G^{(t)}}(v) | + | E^{(t)}(u,v) | + | E(G[N_{G^{(t)}}(u) \cap N_{G^{(t)}}(v)]) |$ and the rule is only applied on pairs of nodes whose distance is at most $2$, it is feasible to find a value $S(c)$ such that the process with parameters $\alpha=\beta=2$ and $G^{(0)} = G_{S(c)}$ has a cycle size of at least $c$.
\end{theorem}
\begin{proof}
Due to Lemma~\ref{lem:Common_Counterexamples}, it suffices to prove this for the process with energy definition $\mathcal{E}(u,v)=| N_{G^{(t)}}(u) \cap N_{G^{(t)}}(v) |$, as long as $S$ is odd and not divisible by $3$. If $2^{\lceil \frac{c}{2} \rceil }+1$ is not divisible by $3$, then our theorem holds for $S(c)=2^{\lceil \frac{c}{2} \rceil }+1$, as in Theorem~\ref{thm:Common_Neibs_Cycle}. Else, $2^{\lceil \frac{c}{2} \rceil }+1=3z_1$ for some integer $z_1$. We set $S(c)=2^{\lceil \frac{c}{2} \rceil }-1=3z_1-2$, which is not divisible by $3$. We know that the process has a cycle size of $2k$, where $k>0$ is the smallest integer such that $2^k \equiv \pm 1 \mod S(c)$, due to Lemma~\ref{lem:Cycle_Size}. For $k=\lceil \frac{c}{2} \rceil$ we have that $2^k \equiv 1 \mod S(c)$. All values $t$, where $0<t\leq k-1$ have $1<2^{t}\leq \frac{2^k}{2}=\frac{S(c)+1}{2}<S(c)-1$, and thus it holds that $2^{t} \not\equiv \pm 1 \mod S(c)$. The cycle size is $2\lceil \frac{c}{2} \rceil \geq c$.
\end{proof}

\section{Computing Shortest paths} \label{sec:shortest}

In this section, we follow a rather different approach than the one in Section~\ref{sec:min}. The definition of energy is based on the sequential algorithms for the corresponding shortest paths problem. In this we want to present the expressive power of the restricted network system. More specifically, we solve the following two problems:
\begin{itemize}
	\item Single-Source All-Destination Shortest Paths (SSAD-SP): We are given an initial connected graph $G^{(0)}$, and a specified node $source$ (all other nodes are considered anonymous). We describe a network system converging to a graph that contains all edges belonging to some shortest path whose one endpoint is $source$, and doesn't contain any other edge.
	\item Single-Source Single-Destination Shortest Paths (SSSD-SP): We are given an initial connected graph $G^{(0)}$, and two specified nodes $source$ and $target$ (all other nodes are considered anonymous). We describe a network system converging to a graph that contains all edges belonging to some shortest path between $source$ and $target$ , and doesn't contain any other edge (except for an edge between $source$ and $target$ which signals the convergence of the process).
\end{itemize}	
	
Notice that, due to the nature of the problems, it is necessary to distinguish between the $source$, the $target$, and the rest of the nodes. For this reason, we have $\forall{t}$ that $K^{(t)}(u)=1$ if $u=source$, $K^{(t)}(u)=2$ if $u=target$, and $K^{(t)}(u)=3$ otherwise. However, for clarity, we just write statements of the form "If $u=source$" instead of "If $K^{(t)}(u)=1$".

Our approach for the first problem is similar to a $BFS$ (Breadth-First Search). For the second problem, the previous approach is not enough. This is because, when a $BFS$ from $source$ reaches $target$, not all edges left are part of shortest paths between the two nodes. To overcome this difficulty, we notice that when a $BFS$ from $source$ reaches $target$, the edges visited by the $BFS$ which are also incident to $target$ (distance $1$) are indeed part of the desired output. We mark them (in a way to be specified later), and start a new $BFS$. When this one finishes, we are able to mark edges which are a little farther away from $target$ (distance $2$). We repeat until we mark all edges needed.

We note that the network systems we describe in the following subsections are based on local interactions, in the sense that two nodes get connected, only if at the previous time-step they were within some constant distance (say $15$). Also, the energy of an edge takes into account only information from nodes within some constant distance from one of its endpoints. Finally, $C^{(t)}$ excludes only edges around some constant distance (in $G^{(0)}$) from $source$ and $target$.
In what follows, by $d^{(t)}(u,v)$ we denote the shortest distance between nodes $u$ and $v$ at time $t$.

\subsection{Single-Source All-Destination Shortest Paths (SSAD-SP)}
The bird's-eye-view of our approach is the following: A standard BFS would visit nodes having distance from the $source$ $1$, then $2$, then $3$, and so on. The only edges to be removed, are the ones that connect nodes having the same distance from the $source$. We try to recreate this approach; due to the lack of memory in our model, by visiting a node we mean connecting it directly to the source. After everything is finished, we disconnect the source from everything else, except from its initial neighbors.

At each time-step $t>0$, we hold the invariant that $G^{(t)}$ is the same with $G^{(0)}$ for edges far away from $source$ (having an endpoint $u$ where $d^{(0)}(source,u)\geq t+2$). The rest of the edges are exactly the ones we want in our final graph, with the aforementioned exception: $source$ is connected with all nodes $u$ where $d^{(0)}(source,u)\leq t+1$.

Thus, after enough time where there is no node $u$ at distance $d^{(0)}(source,u)\geq t+2$, we almost achieve the desired graph, with the difference being that $source$ is connected with every node. At this point, we drop all edges with one endpoint being $source$. The set $C^{(t)}$ keeps only those edges that were present at $t=0$. Finally, we take extra care in order for the process to converge, and not start all over after this point, connecting $source$ with nodes that shouldn't be connected to it.

More specifically, we can pick any $\alpha$ and $\beta$, as long as $0 < \alpha < \beta$. $C^{(t)}$ contains all edges except those incident to $source$ at time $0$, for all $t$. For compactness, in the definition of energy, we skip symmetric cases; that is, swapping $u$ and $v$ gives the same energy value, even if not stated explicitly.

\[   
\mathcal{E}^{(t)}(u,v) = 
     \begin{cases}
     	0							& if~ t=0, d^{(t)}(source,u) = d^{(t)}(source,v) = 1 \\
		\beta						& if~ t=0, u=source, d^{(t)}(source,v)=2 \\
		0							& if~ d^{(t)}(source,u) = d^{(t)}(source,v) = 2 \\
		0							& if~ u=source, d^{(t)}(source,v)=1, \not\exists w~where~d^{(t)}(source,w)=2 \\
		\beta						& if~ t>0, u=source, d^{(t)}(source,v)=2, u~is~part~of~some~triangle \\
		\alpha						& otherwise \\
     \end{cases}
\]

We refer to these six branches of the energy definition as cases (1), (2)... (6).

\begin{theorem}
There exists a network system that solves the SSAD-SP problem.
\end{theorem}
\begin{proof}
We argue that the above definitions of $\alpha, \beta, C^{(t)}, \mathcal{E}^{(t)}(u,v)$ define a network system that solves the SSAD-SP problem. Let us call an edge $(u,v)$ $k-good$ if it is part of the desired output, and $d^{(0)}(source,u)\leq k$, $d^{(0)}(source,v)\leq k$. Similarly, $k-bad$ edges are not part of the desired output, and $d^{(0)}(source,u)\leq k$, $d^{(0)}(source,v)\leq k$.

We first claim the following:
\begin{claim}
Let $t'$ be the first time-step that $source$ is connected with all nodes. Then, for $G^{(t)}$, where $0<t\leq t'$, it holds that:
	\begin{itemize}
		\item all $(t+1)-good$ edges $(u,v)$ are present
		\item $source$ is connected with all nodes $u$ such that $d^{(0)}(source,u)\leq t+1$, and no other node
		\item no $(t+1)-bad$ edge is present, except for those whose one endpoint is $source$
		\item all other edges are the same as in $G^{(0)}$
	\end{itemize}	
\end{claim}

Supposing that our claim holds, we prove the main statement, and then proceed to prove our claim. First of all, either $t'=0$, or by our claim, $source$ acquires new neighbors at every time step. Thus, $t'$ always exists. If $t'=0$, all edges not incident to $source$ drop (case (1)), and all edges incident to $source$ are preserved by $C^{(t)}$. At subsequent time steps, case (6) and $C^{(t)}$ preserve everything, and thus we reach convergence. Obviously this graph is the desired one, as $source$ was connected with everything at $t=0$.

Else, $t'>0$. Due to our claim, $G^{(t')}$ is the desired graph, except for the edges where one endpoint is $source$. At $t'+1$, all these $(t'+1)-bad$ edges are dropped, due to case (4), and no other change takes place (case (6) and $C^{(t)}$). Thus, we reached the desired graph. To prove that the process converges at this point, notice that at time step $t'+2$, $source$ is not contained in any triangle, as that would violate that $1-bad$ edges were not present at time $t'$. Thus case (5) is never applied. Case (3) is not applied on existing edges, as there were no $2-bad$ edges at time $t'$, and case (4) is not applied as the condition for $w$ is not met. Thus only case (6) is applied, and convergence is reached.

Now, to prove our claim, we notice that it trivially holds if $t'=0$, as there is no $t$ for which $0<t\leq 0$. Thus, initially $source$ is not connected to every node. At $t=1$ all edges connecting nodes with distance from $source$ equal to $1$ are dropped (case (1)), and same for distance $2$ (case (3)). All nodes with distance $2$ are connected to $source$ (case (2)), and all other edges are preserved (case (6)). We conclude that our base-case holds.

Suppose our claim holds for time $t$. All nodes $u$ where $d^{(t)}(source,u)=2$ have $d^{(0)}(source,u)=t+2$, as they have a neighbor $v$ connected with $source$ (by the induction hypothesis this means that $d^{(0)}(source,v)\leq t+1$) and they are not connected to $source$ (by the induction hypothesis $d^{(0)}(source,u)\geq t+2$). Thus, case (3) implies that no $(t+2)-bad$ edge is present in $G^{(t+1)}$, except for those whose one endpoint is $source$. Case (5) implies that $source$ is connected with all nodes $u$ such that $d^{(0)}(source,u)\leq t+2$, and no other node. Case (6) implies that all $(t+2)-good$ edges $(u,v)$ are present in $G^{(t+1)}$, and all other edges are the same as in $G^{(0)}$.
\end{proof}

\subsection{Single-Source Single-Destination Shortest Paths (SSSD-SP)}
A standard algorithm for this problem is that of Listing~\ref{lst:sssd-sp-linear}. In our case, we can't recreate this approach, because it is not clear to us how to achieve the loop which scans nodes in decreasing distance from the $source$.

\begin{lstlisting} [breaklines=true, caption={SSSD-SP Linear Time Solution},label=lst:sssd-sp-linear,float,abovecaptionskip=-\medskipamount]
Run a BFS from source, until you reach target,
	deleting edges between nodes u,v if d(source,u)=d(source,v)
Mark target
for (i=d(source,target)-1; i>=0; i--)
	Let S(i) be the set of nodes u with d(source,u)=i
	Mark nodes of S(i) which have a marked neighbor
For all edges, keep only those whose both endpoints are marked
\end{lstlisting}

Instead, we modify the algorithm of Listing~\ref{lst:sssd-sp-linear} to the slower algorithm of Listing~\ref{lst:sssd-sp-quadratic}. The algorithm repeatedly runs $BFS$ from $source$. Each $BFS$ stops when it reaches a marked node and marks all visited nodes that can reach marked nodes.

\begin{lstlisting} [breaklines=true, caption={SSSD-SP Quadratic Time Solution},label=lst:sssd-sp-quadratic,float,abovecaptionskip=-\medskipamount]
Mark target
for (i=d(source,target)-1; i>=0; i--)
	Run a BFS from source, until you reach a marked node
	deleting edges between nodes u,v if d(source,u)=d(source,v)
	Mark visited nodes which have a marked neighbor
For all edges, keep only those whose both endpoints are marked
\end{lstlisting}

On a higher level view, we need $4$ different types of marks to implement the above algorithm. First of all, we mark nodes visited during a $BFS$. This is simulated by a connection with $source$, but not with $target$. We also mark nodes so that we later delete edges that do not have both endpoints marked. This is simulated by a connection with $target$, but not $source$. Additionally, when all desired edges have been found, no new $BFS$ should start. This is simulated by an edge between $source$ and $target$. Finally, we must have a way to "inform" nodes that are far from the shortest path, that they should drop all their edges. To do this, we connect each such node both with $source$ and $target$. At the next step, we drop all the edges of this node, and connect its neighbors with $source$ and $target$. 

More specifically, we can pick any $\alpha$ and $\beta$, as long as $0 < \alpha < \beta$. If $d^{(0)}(source,target)\leq 6$, then $C^{(t)}$ consists of all the edges except the desired ones (which are all within a constant distance from $source$). Else, $C^{(2t)}$ excludes edges incident to $target$ at time $0$, and $C^{(2t+1)}$ excludes edges incident to $source$ at time $0$, for $t\geq 0$. For compactness, we skip symmetric cases; that is, we write some $Condition(u,v)$ but it is implied that it suffices that either $Condition(u,v)~OR~Condition(v,u)$ holds. 

Let us say that a node $u$ is $k-close$ if $d^{(0)}(source,u)\leq k$. An edge $(u,v)$ is $k-good$ if it is part of the desired output, and both its endpoints are $k-close$. Similarly, $k-bad$ edges are not part of the desired output, and both their endpoints are $k-close$. Finally, let $D^{(t)}=d^{(t)}(source,target)$.

To make things more clear, we define $4$ different functions, and combine them to create $\mathcal{E}^{(t)}(u,v)$.

First is the function which describes the behavior at $t=0$. It helps solving the problem in case the distance between $source$ and $target$ is at most $6$. It also removes edges between neighbors of $source$, and edges between neighbors of $target$. This is a preprocessing step, that gives us the ability to make simple assumptions on the more general case (like the fact that $source$ and $target$ should not be connected, as a connection between them is used as a signal to stop any new $BFS$).

\[   
\mathcal{E}_{0}^{(t)}(u,v) = 
     \begin{cases}
     	\beta						& if~ u\in \{source,target\}, D^{(t)}\leq 6, \\
     								& ~~~ d^{(t)}(source,v)=D^{(t)}+2 \\
		\beta						& else~if~ u=source, v=target, D^{(t)} \leq 6 \\
     	0							& else~if~ d^{(t)}(source,u)\leq D^{(t)}+1, D^{(t)}\leq 6 \\
     	0							& else~if~ d^{(t)}(source,u)=d^{(t)}(source,v)=1 \\
     	0							& else~if~ d^{(t)}(target,u)=d^{(t)}(target,v)=1 \\
		\alpha						& otherwise \\
     \end{cases}
\]

We also have a function which simulates the $BFS$.

\[   
\mathcal{E}_{BFS}^{(t)}(u,v) = 
     \begin{cases}
		0							& if~ (t~even~OR~D^{(t)}>4), \\
									& ~~~ d^{(t)}(source,u) = d^{(t)}(source,v) = x~where~x\in \{2,3\} \\
		0							& else~if~ t~odd, u=source, D^{(t)} \in \{3,4\} \\
		\alpha						& else~if~ t~odd, D^{(t)} \in \{3,4\}, u=target, d^{(t)}(u,v)=1, \\
									& ~~~ (d^{(t)}(source,v)<D^{(t)}~OR~target~part~of~triangle) \\
		0							& else~if~ t~odd, D^{(t)} \in \{3,4\}, d^{(t)}(source,u)\in \{D^{(t)},D^{(t)}+1\} \\
		\beta						& else~if~ (t~even~OR~D^{(t)}>4), u=source, d^{(t)}(source,v)=2 \\
		\beta						& else~if~ t~odd, u=target, D^{(t)} \in \{3,4\}, \\
									& ~~~ d^{(t)}(source,v)=D^{(t)}-2, d^{(t)}(target,v)=2 \\
		\beta						& else~if~t~odd, u \in \{source,target\}, D^{(t)} \in \{3,4\}, \\
									& ~~~ d^{(t)}(source,v)=D^{(t)}+2 \\
     \end{cases}
\]

Additionally, we need a function for propagating the information for far-away nodes to drop their edges.

\[   
\mathcal{E}_{Drop}^{(t)}(u,v) = 
     \begin{cases}
		0							& if~ d^{(t)}(source,u) = d^{(t)}(target,u) = 1 \\
		\beta						& else~if~ u\in \{source,target\}, v\not\in \{source,target\}, \\
									& ~~~ \exists w \in N_{G^{(t)}}(v) d^{(t)}(source,w)=d^{(t)}(target,w)=1 \\
     \end{cases}
\]

Finally, we have a function that completes the picture, by making some final modifications, and signaling the stop of future $BFS$.

\[   
\mathcal{E}_{Stop}^{(t)}(u,v) = 
     \begin{cases}
		\beta						& if~ u=source, v=target\\
     	\alpha						& else~if~ d^{(t)}((target,u)=d^{(t)}(target,v)=1~OR~ \\
     								& ~~~ (u,v)~on~some~shortest~path \\
     	0							& otherwise
     \end{cases}
\]

Now we are ready to build our main energy function.

\[   
\mathcal{E}^{(t)}(u,v) = 
     \begin{cases}
     	\alpha							& if~ u~is~not~15-close \\
		\mathcal{E}_{0}^{(t)}(u,v)		& else~if~ t=0, (D^{(t)}\leq 6~OR~source~part~of~triangle \\
										& ~~~ OR~(d^{(t)}(target,u)=1,target~part~of~triangle)) \\
		\mathcal{E}_{Drop}^{(t)}(u,v)	& else~if~ \mathcal{E}_{Drop}^{(t)}(u,v)~specified \\
		\mathcal{E}_{Stop}^{(t)}(u,v)	& else~if~ \mathcal{E}_{Stop}^{(t)}(u,v)~specified, D^{(t)}= 5, \\
										& ~~~ source~not~part~of~triangle \\
		\mathcal{E}_{BFS}^{(t)}(u,v)	& else~if~ \mathcal{E}_{BFS}^{(t)}(u,v)~specified, D^{(t)}>2, \\
									& ~~~ (source~part~of~triangle~OR~D^{(t)}>5) \\
		\alpha							& otherwise
     \end{cases}
\]	

Before proving our main theorem, let us first present some important facts about our process. First of all, not any possible edge can be formed.

\begin{lemma} \label{lem:only_st_edges}
The only edges created, at any given time, are edges where one endpoint is $source$ or $target$.
\end{lemma}
\begin{proof}
Follows directly from the energy definitions, whenever the energy is $\beta$.
\end{proof}

The following Lemma states that if the shortest paths are too small, it is trivial to detect them.

\begin{lemma} \label{lem:Short_Dist_t0}
If it holds that $D^{(0)}\leq 6$, then for all $t>0$, the edges of $(D^{(0)}+1)-close$ nodes, except for $source$ and $target$, are exactly the desired ones.
\end{lemma}
\begin{proof}
At time $t=0$, all $(D^{(0)}+1)-close$ nodes $u$ drop all their $(D^{(0)}+1)-bad$ edges (case (3) of $\mathcal{E}_{0}$), except for the desired ones, which are preserved infinitely long by $C^{(t)}$. Also, an edge between $source$ and $target$ is created (case (2) of $\mathcal{E}_{0}$) and no rule of $\mathcal{E}$ can ever drop an existing $(source,target)$ edge. It suffices to show that no $(D^{(0)}+1)-bad$ edge is created in subsequent steps, connecting a node with either $source$ or $target$, due to Lemma~\ref{lem:only_st_edges}.

This is true because the $(source,target)$ edge means $D^{(t)}=1$ for all $t>0$, and thus $\mathcal{E}_{BFS}$ and $\mathcal{E}_{Stop}$ are never used. For a bad edge to be created, it must be by $\mathcal{E}_{Drop}$; thus $source$ and $target$ get connected with a node $v\not \in \{source,target\}$ which is neighbor of a node $w$ connected both with $source$ and $target$. Since $v$ is assumed to be $(D^{(0)}+1)-close$ and it never creates new edges with nodes different than $source$ or $target$, by Lemma~\ref{lem:only_st_edges}, then $(v,w)$ is part of the shortest path at $t=0$. This means that $w$ is also $(D^{(0)}+1)-close$, or else it wouldn't be connected with $v$. The first time step where a $(D^{(0)}+1)-bad$ edge is created, $(source,w)$ and $(w,target)$ edges are also part of the shortest path. This means that $D^{(0)}= 2$, which implies $(v,w)$ is not part of the shortest path, contradiction. 
\end{proof}

The previous Lemma shows that it is possible to detect the shortest paths, if they are very small. However, we should still drop all edges that do not participate in any shortest path, even if they are very far away from $source$ and $target$. The following Lemma helps us in this direction.

\begin{lemma} \label{lem:Disconnect}
If $G^{(t_0)}$ is the same as $G^{(0)}$, except for edges between $(D^{(0)}+2)-close$ nodes, and all nodes $u$ with $d^{(0)}(source,u) = D^{(0)}+2$ are connected both with $source$ and $target$, but with no other $(D^{(0)}+1)-close$ nodes, then at some time $t'>t_0$ all nodes which are not $(D^{(0)}+1)-close$ become isolated. The rest of $G^{(t')}$ is the same as $G^{(t_0)}$.
\end{lemma}
\begin{proof}
Notice that if at time $0$ an edge is connected to a non $(D^{(0)}+2)-close$ node, then it is not in the neighborhood of $source$ or $target$, nor in some of their shortest paths, and thus it is not contained in $C^{(t)}$, for any $t$.

Due to our assumptions, $D^{(t_0)}\leq 2$. This means that only $\mathcal{E}_{Drop}$ may alter edges. At time $t_0$, all nodes $u$ with $d^{(0)}(source,u)=D^{(0)}+2$ get isolated (case (1) of $\mathcal{E}_{Drop}$) and nodes $v$ with $d^{(0)}(source,v)=D^{(0)}+3$ get connected with both $source$ and $target$ (case (2) of $\mathcal{E}_{Drop}$). This process continues until all non $(D^{(0)}+1)-close$ nodes have become isolated. Let us note that all these nodes are reachable from the $(D^{(0)}+2)-close$ nodes, as $G^{(0)}$ is assumed to be connected.
\end{proof}

The following Lemma presents the ability of our process to simulate the $BFS$ of Listing~\ref{lst:sssd-sp-quadratic}.

\begin{lemma} \label{lem:BFS_sssd}
If $D^{(t_0)}>5$ and $source$ is not part of any triangle, then there is an odd time $t'>t_0$ such that the differences between $G^{(t_0)}$ and $G^{(t')}$ are that $D^{(t')}\in \{3,4\}$ and $source$ is connected to all nodes $u$ with $d^{(t)}(source,u)\leq D^{(t_0)}-D^{(t')}+1$. Furthermore, nodes $v,w$ with $d^{(t_0)}(source,v)=d^{(t_0)}(source,w) \leq D^{(t_0)}-D^{(t')}+2$ are not connected in $G^{(t')}$.
\end{lemma}
\begin{proof}
Due to our assumptions, the only differences between successive graphs occur due to $\mathcal{E}_{BFS}$. As long as $D^{(t)}\geq 5,t\geq t_0$, $source$ gets connected with neighbors of its neighbors (case (5) of $\mathcal{E}_{BFS}$), and edges between nodes with distance $2$ (respectively distance $3$) are dropped (case (1) of $\mathcal{E}_{BFS}$). Notice that $\mathcal{E}_{Stop}$ never occurs, as, even if $D^{(t)}\leq 5$ for some $t>t_0$, $source$ is part of a triangle, consisting of $source$, one of its initial neighbors, and any neighbor of such a neighbor.

The above discussion is formalized as following: At step $t_0+i,i>0$, $source$ is connected with nodes $u$ where $d^{(t_0)}(source,u)\leq i+1$, and there exists no edge $(v,w)$ where $d^{(t_0)}(source,v)=d^{(t_0)}(source,w)\leq i+2$. When $D^{(t)}=4$, either $t$ is odd, and our lemma follows, or it is even, and thus at the next time step it is odd, and our lemma follows.
\end{proof}

We are now ready to prove our main theorem.

\begin{theorem}
There exists a network system that solves the SSSD-SP problem.
\end{theorem}
\begin{proof}
We argue that the above definitions of $\alpha, \beta, C^{(t)}, \mathcal{E}^{(t)}(u,v)$ define a network system that solves the SSSD-SP problem.

First, we study the case where $D^{(0)}\leq 6$. For all $t>0$, the edges of nodes $u$ with $d^{(0)}(source,u)\leq D^{(0)}+1$ are exactly the desired ones, by Lemma~\ref{lem:Short_Dist_t0}. Also, all nodes $u$ with $d^{(0)}(source,u) = D^{(0)}+2$ become disconnected from nodes closer to $source$ (case (3) of $\mathcal{E}_{0}^{(t)}(u,v)$), and get connected with both $source$ and $target$ (case (1) of $\mathcal{E}_{0}^{(t)}(u,v)$). It follows by Lemma~\ref{lem:Disconnect} that all nodes $u$ with $d^{(0)}(source,u)>D^{(0)}+1$ become isolated at some later time, after which they do not form any edge, due to case (1) of $\mathcal{E}^{(t)}(u,v)$. Thus, we converge to the desired graph.

Else, at time $t_0\in \{0,1\}$ it holds that $D^{(t_0)}\geq 6$ and $source$ (respectively $target$) is not part of any triangle (cases (4) and (5) of $\mathcal{E}_{0}^{(t)}(u,v)$). Thus, by Lemma~\ref{lem:BFS_sssd}, we can find an odd time $t'>t_0$ where the differences between $G^{(t_0)}$ and $G^{(t')}$ are that $D^{(t')}\in \{3,4\}$ and $source$ is connected to all nodes $u$ with $d^{(t)}(source,u)\leq D^{(t_0)}-D^{(t')}+1$. Furthermore, nodes $v,w$ with $d^{(t_0)}(source,v)=d^{(t_0)}(source,w) \leq D^{(t_0)}-D^{(t')}+2$ are not connected in $G^{(t')}$. Of course these edges are not part of any shortest path between $source$ and $target$.

At the following time step, as we argued in the proof of Lemma~\ref{lem:BFS_sssd}, changes occur only due to $\mathcal{E}_{BFS}$. As $t'$ is odd, $source$ drops all its edges (case (2) of $\mathcal{E}_{BFS}$) except for those with neighbors at $t=0$, due to $C^{(t)}$. Additionally, $target$ is disconnected by all its neighbors that are not part of a shortest path (cases (3) and (4) of $\mathcal{E}_{BFS}$), and all nodes $u$ with $d^{(t_0)}(source,u)\in \{D^{(t_0)},D^{(t_0)}+1\}$ get disconnected (case (4)). Also, by case (6), $target$ is connected with nodes $v$ for which $d^{(t_0)}(source,v)=D^{(t_0)}-2, d^{(t_0)}(target,v)=2$. These nodes are obviously part of some shortest path. Finally, all nodes $w$ with $d^{(t_0)}(source,w)=D^{(t_0)}+2$ are connected with both $source$ and $target$. By Lemma~\ref{lem:Disconnect}, there is a time $t''>t'$ where the two graphs are the same, with the difference that non $(D^{(0)}-1)-close$ nodes are isolated (except for $target$). These nodes never form any new edge, by case (1) of $\mathcal{E}$.

We notice that $G^{(t'')}$ is the same as $G^{(0)}$ with the following differences: All nodes which are not $(D^{(0)}-1)-close$ (except for $target$) are isolated, all nodes with the same distance from the source at $t=0$ are disconnected from each other, and $target$ is connected with all nodes contained in some shortest path, at distance $2$ from it at $t=0$.

Repeating the same analysis is even easier, because now $target$ is part of a triangle, which means it will not drop any edges, and also $\mathcal{E}_{Drop}$ never occurs again, as there are no nodes left at a proper distance. Thus, there will be an even time $t_1$ where $G^{(t_1)}$ will be the same as $G^{(t'')}$, but $target$ is connected with all nodes contained in some shortest path, at distance $D^{(0)}-4$, and all nodes $u$ which are not part of a shortest path at $t=0$ and $d^{(t_1)}(source,u)\geq 6$ are isolated. Then, $D^{(t_1)}=5$, and only $\mathcal{E}_{Stop}$ occurs. This function connects $source$ and $target$, and drops all $6-bad$ edges and all edges incident to $target$, except for ones it had since $t=0$, due to $C^{(t_1)}$. This means that we reach our desired graph.

To prove convergence, notice that $D^{(t)}=1$ for $t>t_1$, and there is no node connected both with $source$ and $target$, as that would mean the original shortest path had length $D^{(0)}\leq 2$.
\end{proof}

\section{Turing-Completeness}\label{sec:turing}

In this section we describe a local rule under which our Network System is able to simulate Rule $110$, an one-dimensional cellular automaton that Cook proved to be Turing-Complete \cite{DBLP:journals/compsys/000104}. Thus, we prove that there exist local rules for which our Network System is Turing-Complete.

\begin{definition}
Rule $110$ is an one-dimensional cellular automaton. Let $cell^{(t)}(i)$ be the binary value of the $i$-th cell at time $t$. If $cell^{(t)}(i)=0$, then $cell^{(t+1)}(i)=cell^{(t)}(i+1)$. Else, $cell^{(t+1)}(i)$ is $0$ if $cell^{(t)}(i-1)=cell^{(t)}(i+1)=1$, and $1$ otherwise.
\end{definition}

Let $CN^{(t)}(u,v) = |N_{G^{(t)}}(u) \cap N_{G^{(t)}}(v)|$ be the number of common neighbors of $u$ and $v$ at time $t$, and $CE^{(t)}(u,v)=\left| E(G[CN^{(t)}]) \right|$ be the number of edges between the common neighbors of $u$ and $v$ at time $t$. We pick an arbitrary value for $\beta$ and then set $\alpha=\beta$. The energy between $u$ and $v$ is defined as follows:

\[   
\mathcal{E}^{(t)}(u,v) = 
     \begin{cases}
       CE^{(t)}(u,v)+\beta-10	& if~ CN^{(t)}(u,v)=10 ~and~ |E^{(t)}(u,v)|=0	\\
       \beta+12-CE^{(t)}(u,v)	& if~ CN^{(t)}(u,v)=10 ~and~ |E^{(t)}(u,v)|=1	\\
       CE^{(t)}(u,v)+\beta-6	& if~ CN^{(t)}(u,v)=6  							\\
       \beta-1					& otherwise
     \end{cases}
\]

Informally, our simulation of Rule $110$ follows these steps. First, we design a primitive cell-gadget (henceforth $PCG$) that stores binary values, but fails to capture Rule $110$ since it doesn't distinguish between the left and the right cell. Then, by making use of the $PCG$ as a building block, we build the main cell-gadget (henceforth $CG$) that is used to simulate a single cell of the cellular automaton. Finally, each time-step from rule $110$ is simulated using $2$ time-steps of our process; on the first one, some $PCG$s acquire their proper value. On the second step, the rest of the $PCG$s copy the correct value from the ones that already acquired it. Our construction, along with the piecewise energy function allow us to make these two sets of $PCG$s behave differently (some $PCG$s compute the correct value, while the others copy). For clarity purposes, we slightly abuse notation, and consider the time steps of our process to differ by $0.5$ instead of $1$. Thus, we write that the sequence of graphs is $G^{(0)}, G^{(0.5)}, G^{(1)}...$, where graphs $G^{(t+0.5)}$, for $t \in \mathbb{N}$, are transitional states of the graph and have no correspondence with cell states of the cellular automaton.

More formally, a $PCG$ is a pair of nodes $(h_i,l_i)$, such that the existence of an edge between them corresponds to value $1$ and otherwise it correponds to value $0$. This $PCG$ is connected to another $PCG$ $(h_{i+1},l_{i+1})$ by adding all possible edges between these nodes as shown in Figure~\ref{fig:PCG}. In this way, $CE^{(t)}(h_i,l_i)$ would be the sum of values of the two adjacents cell gadgets.

\begin{figure}[t]
\centering
\includegraphics[scale=0.5]{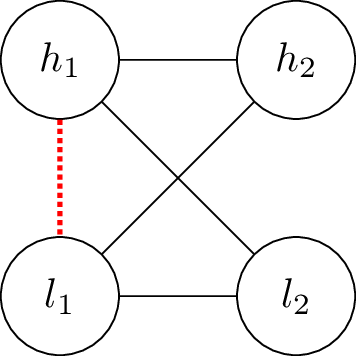}
\caption{We connect $PCG$s $(h_1,l_1)$ and $(h_2,l_2)$ using the $4$ continuous edges. These $4$ edges are never included in any $C^{(t)}$, and thus they never disappear. The dotted edge between $h_1$ and $l_1$ means that the value of $(h_1,l_1)$ is $1$. The value of $(h_2,l_2)$ is $0$.}
\label{fig:PCG}
\end{figure}

The $i-th$ $CG$ that corresponds to the $i$-th cell (we write $CG(i)$) consists of $4$ $PCG$s, which we identify as $A_1(i)$, $A_2(i)$, $B_1(i)$ and $B_2(i)$. We connect each $A_j(i)$ with each $B_k(i)$ ($4$ connections in total, where each connection uses $4$ edges, as depicted in Figure~\ref{fig:PCG}). In order to connect $CG(i)$ (cell $i$) with $CG(i+1)$ (cell $i+1$) we connect $A_j(i)$ with $A_j(i+1)$, and $A_j(i)$ with $B_j(i+1)$, as shown in Figure~\ref{fig:CG}. A $CG$ is said to have value $0$ if all $4$ of its $PCG$s are set to $0$ and $1$ if all $PCG$s are set to $1$. We guarantee that no other case can occur in $G^{(t)}, t\in \mathbb{N}$, although certain cases can occur in the intermediate graphs $G^{(t+0.5)}, t\in \mathbb{N}$.

\begin{figure}[t]
\centering
\includegraphics[scale=0.2]{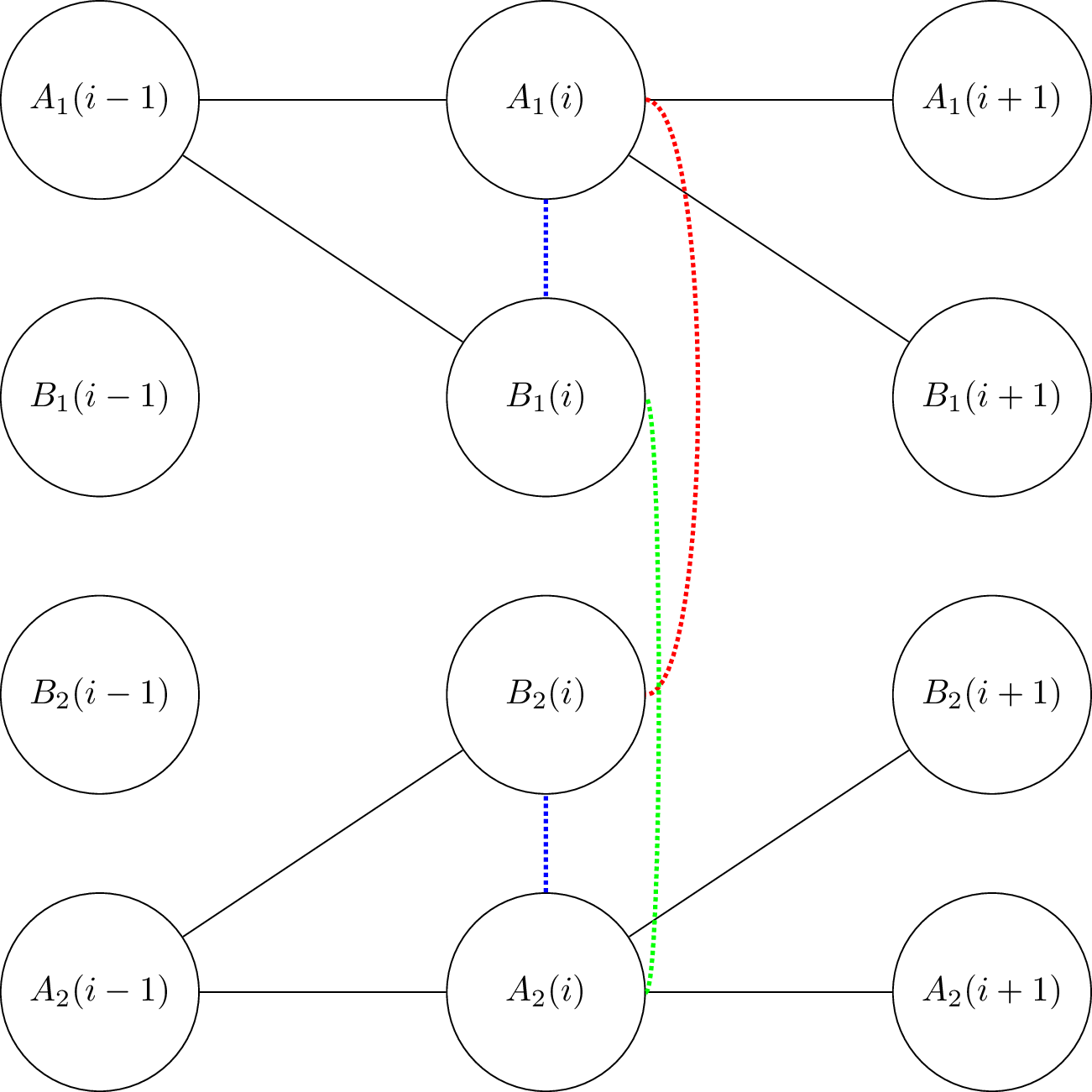}
\caption{Each circle represents a $PCG$ (2 nodes) and each line represents a connection between $PCG$s (4 edges) as in Figure~\ref{fig:PCG}. Only connections relevant to $A_1(i),A_2(i),B_1(i),B_2(i)$ are shown. The $4$ dotted connections in the second column are internal connections of $CG(i)$. All other continuous connections correspond to how $CG(i-1)$ is connected with $CG(i)$ and $CG(i)$ is connected with $CG(i+1)$. None of the edges of these connections is ever included in any $C^{(t)}$, and thus they are always preserved.}
\label{fig:CG}
\end{figure}

Each cell from Rule $110$ is represented by a $CG$, and they are connected by the aforementioned method. At time $t$, where $t\geq 0$ is an integer, we have that $C^{(t)}$ contains all pairs of nodes both belonging in the same $A_j(i)$. In other words, only the edges that define the value of an $A_j(i)$ are allowed to change from time $t$ to time $t+0.5$. Similarly, at time $t+0.5$ we have that $C^{(t+0.5)}$ contains all pairs of nodes both belonging in the same $B_j(i)$, for any valid $i,j$.

We notice that due to the definition of $C^{(t)}$ for any $t$, only edges inside $PCG$s may be allowed to change, meaning that all connections between $PCG$s will remain as is forever. In addition, the number of common neighbors of the pair of nodes $A_j(i)$ is always $CN^{(t)}(A_j(i))=10$, for all valid $t,i,j$, as it has $5$ neighboring $PCG$s, and each $PCG$ has two nodes. Furthermore, it holds that $CE^{(t)}(A_j(i))=8+A_j^{(t)}(i-1)+B_1^{(t)}(i)+B_2^{(t)}(i)+A_j^{(t)}(i+1)+B_j^{(t)}(i+1)$, as the edges between common neighbors are the internal edges of neighboring $PCG$s, plus the connection between $A_j^{(t)}(i-1)$ and $B_j^{(t)}(i)$ ($4$ edges), plus the connection between $A_j^{(t)}(i+1)$ and $B_j^{(t)}(i+1)$ ($4$ edges). Similarly, for a $B_j(i)$ we have that $CN^{(t)}(B_j(i))=6$, and $CE^{(t)}(B_j(i))=4+A_j^{(t)}(i-1)+A_1^{(t)}(i)+A_2^{(t)}(i)$.

\begin{lemma} \label{lem:Simulate_110}
It holds that $A_j^{(t)}(i)=B_j^{(t)}(i)=cell^{(t)}(i)$ for $j\in\{1,2\}$ and all $i,t\in \mathbb{N}$.
\end{lemma}
\begin{proof}
It holds that $A_j^{(0)}(i)=B_j^{(0)}(i)=cell^{(0)}(i)$ by the initialization of our construction. Suppose that $A_j^{(t)}(i)=B_j^{(t)}(i)=cell^{(t)}(i)$ for integer $t\geq 0$. By using induction we show that the lemma holds for time $t+1$.

First of all, we prove that $A_j^{(t+0.5)}(i)=cell^{(t+1)}(i)$. If $cell^{(t)}(i)=0$, then it holds that $cell^{(t+1)}(i)=cell^{(t)}(i+1)=A_j^{(t)}(i+1)=B_j^{(t)}(i+1)$, due to our inductive hypothesis. Furthermore, due to our inductive hypothesis it holds that $A_j^{(t)}(i)=B_1^{(t)}(i)=B_2^{(t)}(i)=0$. Thus, since $CN^{(t)}(A_j(i))=10$ and $|E^{(t)}(A_j(i))|=0$ (there is no edge between the two nodes in $A_j(i)$) the energy between the pair of nodes is $\mathcal{E}^{(t)}(A_j(i))=CE^{(t)}(A_j(i))+\beta-10$. 
To find the energy of the pair of nodes $A_j(i)$ we compute:
\[CE^{(t)}(A_j(i))=8+A_j^{(t)}(i-1)+B_1^{(t)}(i)+B_2^{(t)}(i)+A_j^{(t)}(i+1)+B_j^{(t)}(i+1)=\]
\[8+cell^{(t)}(i-1)+2cell^{(t)}(i+1)\]
Thus, it follows that the energy of $A_j(i)$ is $\mathcal{E}^{(t)}(A_j(i))=\beta+cell^{(t)}(i-1)+2cell^{(t)}(i+1)-2$, which is at least $\beta$ if and only if $cell^{(t)}(i+1)=1$. Thus, in the case where $cell^{(t)}(i)=0$ we proved that indeed it holds that $A_j^{(t+0.5)}(i)=cell^{(t+1)}(i)$.

We use a similar reasoning for the case where $cell^{(t)}(i)=1$. In particular, since $CN^{(t)}(A_j(i))=10$ and $|E^{(t)}(A_j(i))|=1$ (there is an edge between the two nodes in $A_j(i)$) the energy between the pair of nodes is $\mathcal{E}^{(t)}(A_j(i))=\beta +12-CE^{(t)}(A_j(i))$. We compute:
\[CE^{(t)}(A_j(i))=8+A_j^{(t)}(i-1)+B_1^{(t)}(i)+B_2^{(t)}(i)+A_j^{(t)}(i+1)+B_j^{(t)}(i+1)=\]
\[=10+cell^{(t)}(i-1)+2cell^{(t)}(i+1)\]
Thus, it follows that the energy of $A_j(i)$ is $\mathcal{E}^{(t)}(A_j(i))=\beta+2-cell^{(t)}(i-1)-2cell^{(t)}(i+1)$, which is less than $\beta$ if and only if $cell^{(t)}(i-1)=cell^{(t)}(i+1)=1$. This proves that $A_j^{(t+0.5)}(i)=cell^{(t+1)}(i)$.

Furthermore, it trivially holds that $A_j^{(t+1)}(i)=cell^{(t+1)}(i)$. That's because, by definition, $A_j(i) \not\in C^{(t+0.5)}$, and thus $A_j^{(t+0.5)}(i)=A_j^{(t+1)}(i)$. The energy of $B_j(i)$ at time $t+0.5$ is (recall that $CN^{(t)}(B_j(i))=6$):
\[\mathcal{E}^{(t+0.5)}(B_j(i))=CE^{(t+0.5)}(B_j(i))+\beta-6=\beta+2A_j^{(t+0.5)}(i)+A_j^{(t+0.5)}(i-1)-2\]
This is at least $\beta$ if and only if $A_j^{(t+0.5)}(i)=1$, which proves that $B_j^{(t+1)}(i)=cell^{(t+1)}(i)$. 
\end{proof}

The following corollary is a straightforward consequence of this lemma.

\begin{corollary} \label{cor:eq}
It holds that $cell^{(t)}(i)=CG^{(t)}(i)$.
\end{corollary}

We are now ready to prove our main theorem.

\begin{theorem}
The Network System we are studying is Turing-Complete.
\end{theorem}
\begin{proof}
By Corollary~\ref{cor:eq} it follows that Rule $110$ is simulated by the particular network system constructed above. If Rule $110$ converges at step $t$ (meaning that no cell changes state for $t'>t$), then our simulation stabilizes also at time $t+1$ since no change will have taken place in the graph from time $t$ to time $t+1$. Since Rule $110$ is Turing- Complete it follows that the particular Network System is also Turing-Complete.
\end{proof}

As a final note, there are local rules that make the Network System Turing-Complete even if we are not allowed to use $C^{(t)}$, that is if all pairs of nodes are always allowed to create an edge $(\forall t: C^{(t)}=K_n)$. The construction uses properties of cliques to create gadgets that work like always-on edges, as well as gadgets that work like always-off edges. Using the always-on gadgets, we can create new gadgets that work like edges that always flip their status. Of course the energy definition should change accordingly. These are enough to replace $C^{(t)}$ in the above simulation. Although feasible, the construction is too technical and we decided not to include it.

\section{Discussion}\label{sec:discussion}

In this paper we have proposed the Network System as a model for dynamic networks and as a model of computation. We provide some initial results related to the emergent behavior of simple algorithms, its computational capabilities as well as its properties with respect to a certain class of algorithms that change only the topology of the network based on the topology itself alone. 

Our immediate goal for future research is to generalize our results. In particular, we wish to prove also the speed of convergence for the Theorem~\ref{thm:conv_only} as well as provide conditions on the convergence when the algorithm access nodes at distance $>1$, that is not its immediate neighbors only. Another immediate goal is to look at other definitions of energy that give rise to interesting emergent behavior. Note the difference between Sections~\ref{sec:min} and \ref{sec:shortest}. In the former, a simple rule leads to emergent behavior while in the latter we program the network system to acquire a particular result based on our knowledge of the solution in a classical model of computation. This interplay between these two different lines of thinking is of great interest.

This is the reason why we mainly envision NS as a framework to study emergent behavior in a mesoscopic scale where other approaches like statistical physics (macroscopic scale) and dynamical systems theory (microscopic scale) seems not to be able to reach easily meaningful results through theoretical analysis but only through experimental evaluation. Efforts on enriching our understanding on how the microscopic gives rise to the macroscopic has produced several interesting results. For example, a well-known such result is the experiments of Nakagaki et al. \cite{journals/n/NakagakiYT00}, who presented the ability of a slime-mold (Physarum polycephalum) to solve mazes. Later on, researchers have established the validity of the aforementioned claim, and provided more functions that Physarum can compute, from a theoretical point of view \cite{DBLP:journals/corr/BeckerBKKM17, DBLP:journals/corr/abs-1901-07231, DBLP:conf/innovations/StraszakV16}. Bird-flocking is also an intriguing such system, where Chazelle managed to prove convergence \cite{DBLP:journals/jacm/Chazelle14}. Based on these results Chazelle coined the term Natural Algorithms \cite{Chazelle:2012:NAI:2380656.2380679} and argued that traditional mathematics seems to fail to attack such problems in an efficient manner (efficiency refers to expressive power) especially due to the existence of memory within these systems that seems to break any symmetry on which traditional mathematics can be based on. Indeed, the work described in \cite{10.1007/978-3-319-18173-8_1} on Influence Systems is similar to ours in the sense that there is a communication graph that governs the interactions in each time step between agents moving in the euclidean space. In our case, we have the interaction graph but the agents lie on a network and not in the Euclidean space. It would be very interesting to see whether the algorithmic renormalization technique introduced by Chazelle could provide general results with respect to convergence for the NS under mild assumptions. 

As a more long-term objective we are really interested in finding a set of simple local rules such that when combined, we can make meaningful local programs which could be analyzed with respect to properties like convergence as well as with respect to the emergent behavior of the network. This is no small task and may have a great impact to various scientific fields. A more immediate goal towards this direction is to define the NS model for the Physarum polycephalum (or some other physical systems) and give a more combinatorial-like proof of convergence of this system. This means that we will have to move from the global viewpoint given by the systems of differential equations to a local viewpoint expressed by the Network System describing this physical system.


\bibliography{Convergence_Network}

\end{document}